\documentclass[a4paper, 11pt, oneside]{article}
\pdfoutput=1
\usepackage[a4paper, top=30mm, bottom=30mm, left=30mm, right=30mm]{geometry}
\usepackage{setspace} 
\usepackage[T1]{fontenc}
\usepackage{ae,aecompl}
\usepackage[utf8]{inputenc}
\usepackage[english]{babel}
\usepackage[colorlinks=true, citecolor=blue, linkcolor=blue, urlcolor=blue, breaklinks=true]{hyperref}
\usepackage{etoolbox}
\usepackage{color}
\usepackage{xcolor}
\usepackage{url}
\usepackage{natbib}
\usepackage{graphicx, subfigure}
\DeclareGraphicsExtensions{.pdf,.png,.jpg, .eps}
\usepackage{amsmath, amsthm, amssymb, latexsym, bbm}
\usepackage{float}
\usepackage{booktabs}
\usepackage{datetime}
\usepackage{caption}
\usepackage{algpseudocode}

\theoremstyle{plain}
\newtheorem{theorem}{Theorem}

\newtheorem{lemma}{Lemma}
\newtheorem{corollary}{Corollary}
\newtheorem{assump}{}
\newenvironment{assumption}[2]{\renewcommand\theassump{#1}\begin{assump}[#2]}{\end{assump}}
\theoremstyle{definition}

\newtheorem{remark}{Remark}

\newcommand{\set}[1]{\left\{#1\right\}}
\newcommand{\tonde}[1]{ \left( #1 \right)}
\newcommand{\quadre}[1]{ \left[ #1 \right]}
\newcommand{\eps}{\varepsilon}%
\newcommand{\abs}[1]{\left| #1 \right| }

\newcommand{\uno}{\mathbb{I}} 
\newcommand{\Inf}{\infty}

\newcommand{\N}{\mathbb{N}}
\newcommand{\Z}{\mathbb{Z}}
\DeclareMathOperator{\ex}{E}   %
\DeclareMathOperator{\var}{Var}%
\DeclareMathOperator{\cov}{Cov}%
\newcommand{\conv}{\to}

\newcommand{\convp}{\stackrel{\text{p}}{\longrightarrow}}

\newcommand{\ProcEps}{ \{\varepsilon_i\}_{i \in \mathbb{Z}}}
\newcommand{\ProcY}{ \{Y_i\}_{i \in \mathbb{Z}}}

\newcommand{\SNR}{\text{SNR}}

\definecolor{ChangeColor}{RGB}{0, 0, 0}     
\newcommand{\REV}[1]{\textcolor{ChangeColor}{#1}}

\definecolor{NOTEColor}{rgb}{1, 0, 1}

\usepackage{enumitem}
\setlist[description]{style=multiline,leftmargin=1.5cm}

\newcommand{\Kern}[1]{\mathcal{K}\left(#1\right)}

\usepackage[ruled,vlined]{algorithm2e}
\SetKwInOut{Data}{data}
\SetKwInOut{Parameter}{parameter}
\SetKwInOut{Input}{input}
\SetKwInOut{Output}{output}

\usepackage{mdframed}
\newmdenv[
  backgroundcolor=gray!20,
  frametitle=,
  skipabove=\topsep,
  skipbelow=\topsep,
]{reminder}

\begin{document}
\thispagestyle{empty}

\title{\Large \MakeUppercase{A fast subsampling method for estimating the distribution of signal-to-noise ratio statistics in nonparametric time series regression models}}

\author{Francesco Giordano \\ Universit\`{a} di  Salerno, Italy \\ giordano@unisa.it
\and %
Pietro Coretto \\ Universit\`{a} di Salerno, Italy \\ pcoretto@unisa.it
}

\date{}

\begin{reminder}
{\centering \color{red} \textbf{\textsf{This is a preprint. The revised version of this paper is published as}}\\}
\vspace{5pt}
F. Giordano and P. Coretto (2020) %
{``A Monte Carlo subsampling method for estimating the distribution of signal-to-noise ratio statistics in nonparametric time series regression models''.} %
\textit{Statistical Methods \& Applications}, %
Vol. 29(3),   %
pp. 483--514. %
(\href{https://doi.org/10.1007/s10260-019-00487-5}{\sf doi: 10.1007/s10260-019-00487-5}).
\end{reminder}
{\let\newpage\relax\maketitle}


\begin{quote}
{\bf Abstract.~}\noindent Signal-to-noise ratio (SNR) statistics play a central role in many applications. A common situation  where SNR is studied is when a continuous time signal is sampled at a fixed frequency with some noise in the background.  While estimation methods exist, little is known about its distribution when the noise is not weakly stationary. In this paper we develop a nonparametric method to estimate the distribution of an SNR statistic when the noise  belongs to a fairly general class of stochastic processes that encompasses both short and long-range dependence, as well as nonlinearities. The method is based on a combination of smoothing and subsampling techniques. Computations are only operated at the subsample level, and this allows to manage the typical enormous sample size produced by modern data acquisition technologies. We derive asymptotic guarantees for the proposed method, and we show the finite sample performance based on numerical experiments. \REV{Finally, we propose an application to electroencephalography (EEG) data.}\\


\end{quote}

\section{Introduction}\label{sec:intro}

Signal-to-noise ratio (SNR) statistics are widely used to describe the strength of the variations of the signal relative to those expressed by the noise. SNR statistics are used to quantify diverse aspects of models where an observable quantity $Y$ is decomposed into a predictable or structural component $s$, often called signal or model, and a stochastic component $\eps$, called noise or error. Although the definition of SNR is rather general in this paper we focus on a typical situation where one assumes a sequence  \REV{$\ProcY$ is determined by
\begin{equation}\label{eq:Yt}
Y_i:= s(t_i) + \eps_i,
\end{equation}
where $i$ is a time index, $s(\cdot)$ is a smooth function of time evaluated  at the time point $t_i$ with $i\in \Z$, and $\ProcEps$ is some random sequence. Assume $t_i \in (0,1)$, however if a time series is observed at time points $t_i \in (a,b)$, these can be rescaled onto the interval $(0,1)$ without changing the results of this paper.}

Equation \eqref{eq:Yt} is a popular model in many applications that range  from physical sciences to engineering, biosciences, social sciences, etc. \citep[see][and references therein]{Parzen1966,Parzen1999}. Although we use the conventional term ``noise'' for $\eps_i$, this term may have a rich structure well beyond  what we would usually consider noise. Some of the terminology here originates from physical sciences where the following concepts have been first explored. Consider a  non stochastic signal $s(t)$ defined on the time interval $(0,1)$, and assume that $s(t)$ has a zero average level (that is $\int_0^1 s(t)dt=0$). The average ``variation'' (or magnitude) of the signal is quantified as  
\begin{equation}\label{eq:Psignal}
P_\text{signal}:= \int_{0}^{1} s^2(t) dt.
\end{equation}
In physical science terminology \eqref{eq:Psignal} is the average power of the signal, that is the ``energy'' contained in $s(\cdot)$ per time unit (if the reference time interval is $(a,b)$ the integral in \eqref{eq:Psignal} is divided by $(b-a)$). If the average signal level is not zero, $s(t)$ is  centered on its mean value, and then  \eqref{eq:Psignal} is computed. The magnitude, or the ``power'', of the noise component is given by $P_\text{noise}:= \var[\eps_i]$. The SNR of the process is the ratio 
\begin{equation}\label{eq:SNRratio}
\SNR:= 10 \log_{10} \frac{P_\text{signal}}{P_\text{noise}},
\end{equation}
expressed in decibels  unit. The SNR can also be defined as the ratio $(P_\text{signal} / P_\text{noise})$, however the decibel scale is more common. Low SNR implies that the strength of the random component of \eqref{eq:Yt} makes the signal $s(\cdot)$ barely distinguishable from the observation of $Y_i$. On the other hand, high SNR means that the  sampling  about $Y_i$ will convey enough information about the predictable/structural component $s(\cdot)$. 

In many analysis, SNR is a crucial parameter to be known. In radar detection applications \citep{Richards2014}, speech recognition \citep{Loizou2013}, audio and video applications of signal processing \citep{Kay1993}, it is crucial to build filtering algorithms that are able to reconstruct $s(\cdot)$ with the largest possible SNR. In neuroscience there is strong interest in quantifying the SNR of signals produced by neurons activity. In fact, the puzzle is that single neurons seem to have low SNR meaning that they emit ``weak signals'' that are still processed so efficiently by the brain system \citep{CzannerSarmaEtal2015}. In medical diagnostics, a physiological activity is measured and digitally sampled (e.g. fMRI, EEG, etc) with methods and devices that need to guarantee the largest SNR possible \citep{UllspergerDebener2010}. The historic discovery of the first detection of a gravitational wave announced on  11 February 2016 has been made possible because of decades of research efforts in designing instruments and measurement methods able to work in an extremely low SNR environment \citep{Kalogera2017}. These are just a few examples of the relevance of the SNR concept. The main goal of this paper is to define a SNR statistic, and provide an estimator for its distribution with proven statistical guarantees under general assumptions on the elements of \eqref{eq:Yt}.

\REV{Let $\mathcal{Y}_n:=\{ y_1, y_2, \ldots, y_n\}$ be sample values of  $Y$ observed at equally spaced time points $t_i=i/n$ for $i=1,2,\ldots,n$, with $t_i \in (0,1)$. That is, in this work we focus on situations where $Y$ is sampled at constant sampling rate (also known as fixed frequency sampling, or uniform design),  although the theory developed here can be extended to non-constant  sampling rates. } Let $\hat{s}(\cdot)$ and $\hat{\eps}$ be estimated quantities based on $\mathcal{Y}_n$. Consider the observed SNR statistic 
\begin{equation}\label{eq:Qn}
\widehat{SNR} := 10 \log_{10} \left( \frac{\frac{1}{n}\sum_{i=1}^{n} \hat s^2\left(\frac{i}{n}\right)}{\frac{1}{m}\sum_{i=1}^m\left(\hat{\eps}_i-\bar{\hat{\eps}} \right)^2} \right), 
\quad \text{with} \quad \bar{\hat{\eps}}=\frac{1}{m} \sum_{i=1}^m \hat{\eps}_i,
\end{equation}
\REV{for some choice of  an appropriate  sequence $\{m\}$ such that \REV{$m \to \infty$ and $m/n \to 0$ as $n \to \infty$}. In this paper we propose a subsampling strategy that consistently estimates the quantiles of the distribution  of $\tau_m (\widehat{SNR} - \text{SNR})$ for an appropriate sequence $\{\tau_m\}$ (see Theorem \ref{th:convergence_dist_Q}). These quantiles are used to construct simple confidence intervals for the   SNR parameter.}

In most applications the observed $\ProcY$ is treated as ``stable'' enough   so that  smoothing methods (typically linear filtering) are applied to get $\hat{s}(\cdot)$ and the error terms, $\hat{\eps}_i$, $i=1,2,\ldots,n$. Therefore, the general practice is to divide the observed data stream into sequential blocks of overlapping observations of some length  (time windowing), and for each block  $\widehat{SNR}$ is computed \citep[see][]{Kay1993, Weinberg2017}.  These SNR measurements are then used to construct its distribution to make 
inference statements about the underlying SNR. These windowing methods implicitly assume some sort of local stationarity and uncorrelated noise, but there is a lack of theoretical justification. However, data structures often exhibit strong time-variations and other complexities not consistent with these simplifying assumptions.  To our knowledge, a framework and a method for estimating the distribution of SNR statistics like \eqref{eq:Qn} with provable statistical guarantees does not exist in the literature.

\REV{The major contribution of this paper is a subsampling method for the approximation of the quantiles of  the distribution of the centered statistic $\tau_m (\widehat{SNR} - \text{SNR})$. The method is based on the observation that $\tau_m (\widehat{SNR} - \text{SNR})$ can be decomposed into the sum of  the following components:
\begin{equation}\label{eq:root_decomp_intro}
 -10\tau_m\left[ %
 \underbrace{\log_{10}\left(1+\frac{\hat{V}_m-\sigma_{\eps}^2}{\sigma_{\eps}^2}\right)}_{\text{error contribution}} %
 -\underbrace{\log_{10}\left(1+\frac{\frac{1}{n}\sum_{i=1}^{n} \hat s^2(i/n)-\int s^2(t) dt}{\int s^2(t) dt}\right)}_{\text{signal contribution}}\right],
\end{equation}
where %
$\hat{V}_m  = m^{-1}\sum_{i=1}^m\left(\hat{\eps}_i-\bar{\hat{\eps}} \right)^2$. In \eqref{eq:root_decomp_intro} the two components reflect the power contribution of the signal estimated based on $\hat{s}^2(\cdot)$, and the power contribution of  and the error term estimated in terms of  $\hat{V}_m$. In practice, the proposed method, formalized in Algorithm \ref{algo}, works as follows:  
(i) the observed time series is randomly divided into subsamples, that is random  blocks  of consecutive observations; (ii) in each subsample  the estimates $\hat{s}^2(\cdot)$ and $\hat{V}_m$ in \eqref{eq:root_decomp_intro} are computed; (iii) finally these subsample estimates are used to approximate the distribution of $\tau_m (\widehat{SNR} - \text{SNR})$ and its quantiles.}

Based on \cite{altman-1990}, a consistent kernel smoother with an  optimal bandwidth estimator is derived. The smoother does not require any further tuning parameters, even though the  stochastic  structure here is richer than that considered in the original paper by  \cite{altman-1990}.  The subsampling procedure extends the contributions of \cite{politis-romano-1994} and \cite{politis-etal-1999}. The main difference with the classical subsampling is that the method proposed in this paper does not require computations on the entire observed time series. Therefore the kernel smoothing is performed subsample-wise, and this  is particularly beneficial in applications when the sample size is significantly large. In sleep studies often an electrophysiological record (EEG) of the brain activity is performed in several positions of the scalp, where each sensor samples an electrical signal for  24 hours at 100Hz, implying  $n=8,640,000$ data points for each sensor \citep[see][]{KempZwindermanEtal2000}. Music is usually recorded at $44.1$Khz (ISO 9660), which implies that a stereo song of 5 minutes produces $n=26,460,000$ data points. This  approach has been explored by \cite{CorettoGiordano2016x} for the estimation of the dynamic range of music signals. However, the work of  \cite{CorettoGiordano2016x} deals with noise structures less general than those studied here. A further original element  of this work is that, although the setup for $\ProcEps$ does not exclude long memory regimes, the methods proposed do not require the identification of any long memory parameter.

The rest of the paper is organized as follows. In Section \ref{sec:model_assumption} we define and discuss the reference framework for $\ProcY$. In Section \ref{sec:estimation} the main estimation Algorithm \ref{algo} is introduced.  The smoothing step of Algorithm \ref{algo} is studied in Seciton \ref{sec:smoothing}, while the subsampling step is investigated in Section \ref{sec:subsampling}. In Section \ref{sec:numeric} we show finite sample results of the proposed method based on simulated data, moreover an application to real data is illustrated. Final remarks and conclusions are given in Section \ref{sec:conclusions}.  All proofs are given in the final Appendix.

\section{Setup and assumptions}\label{sec:model_assumption}

The framework \eqref{eq:Yt} underpins  a popular strategy to study experiments where a continuous time (analog)  signal $s(\cdot)$ is sampled at fixed time points $t_i$. The stochastic term $\ProcEps$ represents various sources of randomness. In some cases, the source and the structure of the random component are known, but this does not apply universally. A ubiquitous assumption about $\ProcEps$ is that it is white noise; sometimes the simplification pushes further towards Gaussianity \citep[see][and references therein]{Parzen1966}. However, in various applications the evidence of departure from this simplicity it is quite rich.

The most elementary source of randomness is the quantization noise, i.e. the added noise introduced by the quantization of the signal.  In Music, speech, EEG and many other applications, a voltage amplitude is recorded at fixed time intervals using a limited range of integer numbers. This is the so called Pulse Code Modulation (PCM), which is at the base of digital encoding techniques. The quantization noise is produced by the rounding error of the PCM sampling. Theoretically the quantization noise is a uniform white noise process, however, \cite{Gray1990} showed that the structure of the quantization noise varies a lot across applications and measurement techniques, and often the white noise assumption is too restrictive. Apart from quantization noise, the recorded signal may be affected by a number of disturbances unrelated to the signal. Take for example an EEG acquisition where electrical noise from the power lines is injected into the measuring device. In speech recording microphones capture stray radio frequency energy. Another example is that of wireless signal transmission affected by multi-path interference, that is: waves bounce off of and around surfaces creating unpredictable phase distortions. Complex effects like these happen in radar transmission too, where it is well known that the Gaussian white noise assumption is generally violated \citep[see][and references therein]{ConteMaio2002}.

Sometimes the stochastic component does not only include unpredictable external artifacts. There are cases where the structure of $\ProcEps$ is the result of several complex phenomena occurring within the system under study. In their pioneering works \cite{Voss_Clarke_1975, Voss_Clarke_1978}  found evidence of $1/f$--noise or similar fractal processes in recorded music. Similar evidence is documented in  \cite{LevitinChordiaEtal2012}.   $1/f$--noise is a  stochastic process where the spectral density follows the power law $c|f|^{-\beta}$, where $f$ is the frequency,  $\beta$ is the exponent, and $c$ is a scaling constant.  $\beta=1$ gives pink noise, that is just an example of such processes.  Depending on $\beta$ these forms of noise are characterized by slowly vanishing serial correlations and/or what is known as long memory. Many electronic devices found in data acquisition instruments introduce $1/f$-type noise \citep{Kogan1996, Weissman1988, KenigCross2014}.  Evidence of departure from linearity and Gaussianity in the transient components of music recordings was also found in \cite{brillinger-irizarry-1998}  and \cite{CorettoGiordano2016x}.

The main goal of this paper is to build an estimation method for the distribution of the SNR that works under the most general setting. Of course achieving universality is impossible, but here we set a model environment that is as rich as possible. Our model is restricted by the following assumptions:

\begin{assumption}{\bf A1}{}\label{ass:s}
The function $s(\cdot)$ has a continuous second derivative.
\end{assumption}

\begin{assumption}{\bf A2}{}\label{ass:eps}
The sequence $\ProcEps$ fulfills one of the following: 
\begin{description}

\item [\normalfont (SRD)]  $\ProcEps$ is a strictly stationary and $\alpha$-mixing process with mixing coefficients $\alpha(k)$, $\ex[\eps_i]=0$,  $\ex \abs{\eps_i^2}^{2+\delta} < +\Inf$, and $\sum_{k=1}^{+\Inf} \alpha^{\delta/(2+\delta)}(k)<\Inf$ for some $\delta>0$.

\item [\normalfont (LRD)] $\eps_i=\sum_{j=0}^{\infty}\psi_ja_{i-j}$ with $\ex[a_i]=0$, $\ex[a_i^4]<\infty\quad \forall i$, $\{a_i\}\sim i.i.d.$, $\psi_j\sim C_1j^{-\eta}$ with $\eta=\frac{1}{2}(1+\gamma_1)$, $C_1>0$ and $0<\gamma_1\le 1$.
\end{description}
\end{assumption}

Assumption \ref{ass:s} reflects a common smoothness requirement for $s(\cdot)$ which does not need further discussion. In most applications,  $s(\cdot)$ will represent the sum of possibly many harmonic components, or long term smooth trends. \ref{ass:eps} sets a wide range of possible structures for the stochastic component. Two regimes are considered here: short range dependence (SRD) and long range dependence (LRD).  SRD is a rather general $\alpha$-mixing assumption that  allows to overcome the usual linear process assumption. The latter  is essential to model fast decaying energy variations that is typical in some form of noise.  Assumptions \ref{ass:s} and \ref{ass:eps}-SRD are also considered in \cite{CorettoGiordano2016x} for the estimation of the dynamic range of music signals. However, in this paper, we are interested in SNR statistics, and we extend the analysis to the cases where LRD occurs. \ref{ass:eps}-LRD has the role to capture situations where the noise spectra shows long-range dependence; in practice this assumption accommodates the  $1/f$-type noise. The  LRD is controlled by  $\gamma_1$ which is between zero and 1. Note that \ref{ass:eps}-LRD assumes a linear structure while \ref{ass:eps}-SRD does not. SRD assumption allows for dependence, and the rate at which it vanishes it is controlled by $\delta$. Under SRD, in the infinite future, the terms of $\ProcEps$ act as an independent sequence. Hence SRD can capture many different forms of dependence but not long memory features. Is the linearity structure of LRD a strong assumption for the long-memory cases? The class of long-memory linear processes is well known in the literature, and in most cases,  LRD effects are more common to appear with a linear autoregressive structure. Moreover, \ref{ass:eps}-LRD  is compatible with the classical parametric models for LRD, e.g. the well known ARFIMA class, already used to capture the $1/f$-noise phenomenon. One could overcome the linearity assumption in LRD but at the expense of serious technical complications. It is important to stress that we are not interested in identifying SRD-vs-LRD, and we want to avoid the additional estimation of the LRD order. The latter is crucial in most parametric models for LRD. Assumption \ref{ass:eps} only defines plausible stochastic structures that can occur in the most diverse applications. Note that \ref{ass:eps}-LRD does not imply that $\ProcEps$ is a Gaussian process or a function of it, as it is assumed in \cite{Jach-etal-2012} and \cite{Hall-b-etal-1998}.

\section{The  smoothing-subsampling procedure}\label{sec:estimation}

\begin{algorithm}[!t]
\caption{blockwise smoothing}\label{algo}
\SetAlgoLined
\Input{data $\{y_1, y_2, \ldots, y_n\}$, constants $K \in \N$, $b \in \N$, and  $b_1 = o\left(b^{4/5}\right)$}
\Output{quantiles of the SNR statistic}
\BlankLine
\BlankLine
Draw without replacement and with uniform probability  a random sample $\mathcal{T}_K = \set{t_1, t_2, \ldots, t_K}$  from the set  $\set{1,2,\ldots, n-b+1}.$\\
\BlankLine
\BlankLine
\For{$t \in \mathcal{T}_K$}{
\begin{itemize}
\item Consider the subsample $\mathcal{Y}_t=\set{y_t, y_{t+1}, \ldots, y_{t+b-1}}$. Based on kernel methods estimate  $s(\cdot)$, and the signal power 
\REV{
\begin{equation}\label{eq:algo_sig_power}
\hat{U}_{n,b,t} = \frac{1}{b}\sum_{i=t}^{t+b-1} \left[\hat{s}\left(\frac{i-t+1}{b} \right)\right]^2.
\end{equation}
}
\item Compute $\hat{\eps}_{i} = y_{i}- \hat{s}((i-t+1)/b)$ for $i=t,t+1,\ldots,{t+b_1-1}$.\\

\item \REV{Estimate the noise variance on $b_1$ values
\begin{equation}\label{eq:algo_noise_var}
\bar{\hat{\eps}}_{b_1,t} = \frac{1}{b_1}\sum_{i=t}^{t+b_1-1}\hat{\eps}_i, \qquad 
\hat{V}_{n,b_1,t} = \frac{1}{b_1} \sum_{i=t}^{t+b_1-1} (\hat{\eps}_i-\bar{\hat{\eps}}_{b_1,t})^2.
\end{equation}}

\item Compute the subsample SNR statistic
\begin{equation}\label{eq_algo_snr}
\widehat{SNR}_{n,b,t} = 10\log_{10}\left(\frac{\hat{U}_{n,b,t}}{\hat{V}_{n,b_1,t}}\right).
\end{equation}
\end{itemize}
}
\BlankLine
Based on   $\{\widehat{SNR}_{n,b,t},\}_{t \in \mathcal{T}_K}$ quantiles  are computed as in  \eqref{eq:SNRquantiles}
\end{algorithm}
The SNR distribution is estimated performing Algorithm \ref{algo}.  This is a simple smoothing-subsampling procedure where for each subsample  $P_{\text{signal}}$ is consistently estimated by  $\hat{U}_{n,b,t}$, and  $P_{\text{noise}}$ is estimated by $\hat{V}_{n,b_1,t}$ on a secondary subsample taken from the previous one. Details and theoretical motivation of the procedure will be treated in Sections \ref{sec:smoothing} and \ref{sec:subsampling}. 
The distribution is constructed for the SNR expressed in decibel scale. 

\REV{The procedure is called ``Monte Carlo'' because the subsample is selection is randomized.  The latter reduces the huge number of subsamples to be explored.} Note that here none of the calculations involve computations over the entire observed sample $\mathcal{Y}_n$. The latter differs  from the classical subsampling for time series data introduced in  \cite{politis-romano-1994} and \cite{politis-etal-1999}. In the classical subsampling, one would estimate the variance of $\ProcEps$ based on the entire sample. This would require that the estimation of $s(\cdot)$ is performed globally on $\mathcal{Y}_n$. In Algorithm \ref{algo} both $s(\cdot)$, and the variance of $\ProcEps$ in \eqref{eq:algo_noise_var} are estimated blockwise. This blockwise smoothing strategy, where computations are performed only at the subsample level, has been  proposed in \cite{CorettoGiordano2016x}. The advantages over the classical subsampling are twofold. First, thanks to the increased data acquisition technology, in most of the applications mentioned in Section \ref{sec:intro}, $n$ scales in terms of millions or billions of data points. It is well known that kernel and other nonparametric smoothing methods become computationally intractable for such big sample sizes. In Algorithm \ref{algo} the computational complexity for the calculation of $\hat{s}(\cdot)$ is governed by the subsample size $b$, which is chosen much smaller than $n$ (see Theorem~\ref{th:conv_var_dist}). Second, the kind of signals we want to reconstruct may exhibit strong structural variations along the time axis,  therefore,  estimation of  $s(\cdot)$ on the entire sample would require the use of optimal kernel methods with local bandwidth increasing the computational burden even more. Working on smaller data chunks allows treating the signal locally. Therefore, simpler kernel methods based on global bandwidth within the subsampled block are better suited to capture the local structure of the signal.  Optimal estimation of $s(\cdot)$, and the random subsampling part of Algorithm \ref{algo} are developed in the next two Sections.


\section{Optimal signal reconstruction}\label{sec:smoothing}

Unless one has enough information about the shape of $s(\cdot)$, nonparametric estimators of functions with proven statistical properties are natural candidates to reconstruct the underlying signal. Our choice is the classical Priestley-Chao kernel estimator \citep[][]{priestley-chao-1972}, because it can be easily optimized in regression models where the error is not necessarily uncorrelated. The estimator for $s(\cdot)$ is defined as 
\begin{equation}\label{eq:ker}
\hat{s}(t) = \frac{1}{nh} \sum_{i=1}^n \Kern{\frac{t-i/n}{h}}y_i.
\end{equation}
The following assumption involving the kernel function $\Kern{\cdot}$ and the bandwidth $h$ is assumed to hold.
\begin{assumption}{\bf A3}{}\label{ass:K}
$\Kern{\cdot}$ is a  density function with compact support and symmetric about $0$. Moreover, $\Kern{\cdot}$ is Lipschitz continuous of some order. The bandwidth $h \in H=[c_1\Lambda_n^{-1/5}, \; c_2\Lambda_n^{-1/5}]$, where $c_1 < c_2$ are two positive constants such that: $c_1$ is arbitrarily small, $c_2$ is arbitrarily large. Define
\begin{equation}\label{Lambda} 
\Lambda_n:=%
\begin{cases}
  n                & \text{if \ref{ass:eps}-SRD holds,} \\
\frac{n}{\log n}   & \text{if \ref{ass:eps}-LRD holds with} \quad  \gamma_1=1,\\ 
n^{\gamma_1}       & \text{if \ref{ass:eps}-LRD holds with} \quad   0<\gamma_1<1.
\end{cases}
\end{equation}
Whenever $n \conv \infty$ it happens that $h \conv 0$ and $\Lambda_nh \conv \Inf$.
\end{assumption}
There are a number of possible choices for $\Kern{\cdot}$ satisfying  \ref{ass:K}, and  we will  use the Epanechnikov kernel for its well known efficiency properties. Setting an optimal bandwidth in \eqref{eq:ker} when the error term may be correlated requires special care. Here an optimal choice of $h$ is even more involved due to the fact that $\ProcEps$ may follow either the SRD or the LRD regime. The sequence \eqref{Lambda} has a role in managing this added complexity. \cite{altman-1990} developed the Priestley-Chao kernel estimator \eqref{eq:ker} with dependent additive errors, and showed that under serial correlation  standard bandwidth optimality theory does not apply. \cite{altman-1990} proposed to estimate an optimal $h$ based on a cross-validation function accounting for the dependence structure of $\ProcEps$. Altman's contribution deals with errors  belonging to the class of linear processes with finite memory. Therefore, Altman's assumptions do not allow the LRD case. Moreover, we consider the SRD assumption because it is typical for stochastic processes with a nonlinear model representation in time series framework. Finally, \cite{altman-1990} assumes that the true autocorrelation function of $\ProcEps$ is known which is not the case in real world applications. 

Let $\hat{\eps}_i = y_i - \hat{s}(i/n)$, and define the cross-validation objective function
\begin{equation}\label{eq:CV}
\text{CV}(h)= %
\quadre{1-\frac{1}{nh}\sum_{j=-M}^M \Kern{\frac{j}{nh}} \hat{\rho}(j)} ^{-2} %
\frac{1}{n} \sum_{i=1}^n \hat{\eps}_i^2.%
\end{equation}
The optimal bandwidth is estimated by minimizing \eqref{eq:CV}, that is 
\begin{equation*}
\hat h = {\text{argmin}}_{h \in H} \; \text{CV}(h).
\end{equation*}
The first term in \eqref{eq:CV} is the  correction factor proposed by \cite{altman-1990},  but replacing the true unknown autocorrelations with their sample counterparts $\hat{\rho}(\cdot)$ up to the $M$th order. $M$ is an additional smoothing parameter, but  Altman's contribution does not deal with its choice. Consistency of the optimal bandwidth estimator is obtained if $M$ increases at a rate smaller than the product $nh$. As in \cite{CorettoGiordano2016x}  $M$ is chosen so that the following holds.
\begin{assumption}{\bf A4}{}\label{ass:M}
Whenever $n\to \Inf$; then $M \to \infty$ and $M=O(\sqrt{nh})$.
\end{assumption}
Let $\text{MISE}(\hat{s};h)$ be the mean integrated square error of $\hat{s}(\cdot)$, that is 
\[
\text{MISE}(\hat s; h) = \int_0^1 \text{MSE}(\hat{s}(t);h)\;dt
\quad \text{where} \quad 
\text{MSE}(\hat{s}(t);h) = \text{E}[(\hat{s}(t) - s(t))^2].
\]
Let $h^\star$ be the global minimizer of $\text{MISE}(\hat{s};h)$. The next result states the optimality of the kernel estimator.
\begin{theorem}\label{th:kern_opt}
Assume {\ref{ass:s}}, {\ref{ass:eps}}, {\ref{ass:K}} and {\ref{ass:M}}.  $\hat{h}/{h^\star} \convp 1$ as $n \to \Inf$.
\end{theorem}
The previous result relates $\hat h$ to the optimal global bandwidth for which convergence rate is known, that is  $O(\Lambda_n^{-1/5} )$.   Theorem~\ref{th:kern_opt} is equivalent to that given in  \cite{CorettoGiordano2016x}, however, the difference here is that $\ProcEps$ may well follow LRD. Therefore, proof of Theorem~\ref{th:kern_opt} (given in the Appendix) needs some further developments.   
\begin{remark} Theorem~\ref{th:kern_opt} improves the existing literature in several aspects. First of all, the proposed signal reconstruction is optimal (in the MISE sense) under both SRD and LRD. Its key feature is that one does not need to identify the type of dependence, that is  SRD vs LRD. There are only two smoothing tunings: $h$ that is estimated optimally, and $M$ fixed according to \ref{ass:M}. The  SRD regime is already treated in \cite{CorettoGiordano2016x}. Regarding LRD,  the result should be compared to  \citet{Hall-a-etal-1995}. The advantages of our approach compared to the latter are: (i) the method is simplified by  eliminating a tuning needed to deal with LRD,  that is the block length for the leave-$k$-out cross-validation in \cite{Hall-a-etal-1995}. This is because the Altman's cross-validation  correction in \eqref{eq:CV} already incorporates the dependence structure via $\hat{\rho}(\cdot)$, and  $M$ is able to correct \eqref{eq:CV}  without any further step identifying whether LRD or SRD occurs; (ii) here we do not assume existence of higher order moments of $\ProcEps$.
\end{remark}

\section{Monte Carlo approximation of  the subsampling distribution}\label{sec:subsampling}
\REV{In this section we exploit the subsampling procedure underlying Algorithm \ref{algo}. We call this procedure ``Monte Carlo'', because it is based on a random selection of subsamples, and here we provide a Monte Carlo approximation of the subsampling distribution of the statistic of interest.} Let us introduce the following  quantities:
\begin{equation}\label{eq:Vn}
V_n=\frac{1}{n}\sum_{i=1}^n\left(\eps_i-\bar{\eps} \right)^2, %
\qquad \text{with} \quad \bar{\eps}=\frac{1}{n} \sum_{i=1}^n \eps_i.
\end{equation}
Although the random sequence $\ProcEps$ is not observable, one can work with its estimate. Replace $\eps_i$ with $\hat \eps_i$ in the previous formula and obtain
\begin{equation}\nonumber
\hat V_n=\frac{1}{n}\sum_{i=1}^n\left(\hat \eps_i-\bar{\hat \eps} \right)^2, %
\qquad \text{with} \quad \bar{\hat \eps}=\frac{1}{n} \sum_{i=1}^n \hat \eps_i.
\end{equation}
The distribution of  a proper scaled and centered $\hat V_n$ can now be used to approximate the distribution of  $\tau_n(V_n-\sigma^2_{\varepsilon})$ where $\tau_n$ is defined in (\ref{tau}) and $\sigma_{\varepsilon}^2:=\text{E}[\varepsilon_t^2]$. One way to do this is to perform the  subsampling as proposed in  \cite{politis-etal-1999b} and \cite{politis-etal-1999}.

That is,  for all blocks of observations of length $b$ (subsample size) compute $\hat V_n$. However the number of possible subsample is huge even for moderate  $n$. Moreover, in typical cases where $n$ is of the order of millions or billions of samples, the computation of the optimal  $\hat{s}(\cdot)$  would require an enormous computer power. The problem is solved by performing the blockwise  smoothing of Algorithm \ref{algo} proposed in \cite{CorettoGiordano2016x}.
Therefore, the signal and the average error are estimated block-wise, so that the computing effort is only driven by $b$. This allows making the algorithm scalable with respect to $n$, a very important feature to process data from modern data acquisition systems.  Here we investigate the theoretical properties of the estimation Algorithm \ref{algo}. The formalization is similar to that given in  \cite{CorettoGiordano2016x}, however, here we deal with a different target statistic, and we face the added complexity of the existence of LRD regimes in  $\ProcEps$. 

First define
\begin{equation}\label{tau} 
\tau_n:=
\begin{cases}
n^{1/2}                              &  \text{if \ref{ass:eps}-SRD holds},\\ 
n^{1/2}                              &  \text{if \ref{ass:eps}-LRD holds with} \quad 1/2<\gamma_1\le 1,\\ 
\left(\frac{n}{\log n}\right)^{1/2}  &  \text{if \ref{ass:eps}-LRD holds with} \quad \gamma_1=1/2,\\ 
n^{\gamma_1}                          &  \text{if \ref{ass:eps}-LRD holds with} \quad 0<\gamma_1<1/2.
\end{cases}
\end{equation}
At a given time point $t$  consider a block of observations  of length $b$, and  the statistics computed in Algorithm \ref{algo}:
\begin{equation}\nonumber
V_{n,b,t}=\frac{1}{b}\sum_{i=t}^{t+b-1} (\eps_i- \bar{\eps}_{b,t})^2, %
\qquad \text{and} \qquad
\hat{V}_{n,b,t}=\frac{1}{b}\sum_{i=t}^{t+b-1} (\hat{\eps}_i-\bar{\hat{\eps}}_{b,t})^2,
\end{equation}
with $\bar{{\eps}}_{b,t}=b^{-1}\sum_{i=t}^{t+b-1} {\eps}_i$ and $\bar{\hat{\eps}}_{b,t}=b^{-1}\sum_{i=t}^{t+b-1}\hat{\eps}_i$. The empirical distribution functions of  $\tau_n (V_n -\sigma_{\varepsilon}^2)$, based on the true and estimated noise, respectively, are given by
\begin{align}
G_{n,b}(x)       =& \frac{1}{n-b+1}\sum_{t=1}^{n-b+1}\uno\set{\tau_b\left(V_{n,b,t}-V_n\right)\le x}, \nonumber\\
\hat{G}_{n,b}(x) =& \frac{1}{n-b+1}\sum_{t=1}^{n-b+1} \uno\set{\tau_b(\hat{V}_{n,b,t}-V_n)\le x}. \nonumber
\end{align}
$\uno\set{A}$ denotes the usual indicator function of the set $A$.  $\tau_b$ is defined in \eqref{tau}. Lemma \ref{lemma2} and \ref{prop:2} in the Appendix state that the subsampling based on statistic \eqref{eq:Vn}  is consistent  under both \ref{ass:eps}-SRD and  \ref{ass:eps}-LRD. Notice that results in \cite{politis-etal-1999} can only be used to deal with SRD. The LRD treatment is inspired to \citet{Hall-b-etal-1998} and \citet{Jach-etal-2012}. However, we improve upon their results in the sense that the  Gaussianity assumption for $\eps_t$ is avoided under  \ref{ass:eps}-LRD with $1/2<\gamma_1\le 1$.  The quantiles of the subsampling distribution also converges to the quantiles of the asymptotic distribution of $\tau_n(V_n - \sigma^2_{\eps})$. This  is a consequence of the fact that  $\tau_n (V_n-\sigma_{\varepsilon}^2)$ converges weakly  (see Remark  \ref{rem:distribuzione}). For $ \gamma_2 \in (0,1)$ the quantities  $q(\gamma_2)$, $q_{n,b}(\gamma_2)$ and $\hat{q}_{n,b}(\gamma_2)$ denote respectively the $\gamma_2$-quantiles with respect the distributions $G$,   (see Remark  \ref{rem:distribuzione}), $G_{n,b}$ and $\hat{G}_{n,b}$ respectively.  We adopt  the usual definition that  $q(\gamma_2)=\inf\set{x: G(x)\ge \gamma_2}$. Lemma \ref{cor:1} in the Appendix states the same consistency for the quantiles. The following remark covers the different cases (\ref{ass:eps}-SRD and \ref{ass:eps}-LRD) for the asymptotic distribution of $\tau_n (V_n - \sigma^2_{\eps})$.
\begin{remark}\label{rem:distribuzione}
By \ref{ass:eps} it can be shown that $\tau_n (V_n - \sigma^2_{\eps} )$ converges weakly to a random variable with distribution, say $G(\cdot)$, where $\sigma^2_{\eps} =\ex[\eps_t^2]$. Under \ref{ass:eps}-SRD, $G(\cdot)$ is a Normal distribution. $G(\cdot)$ is still a Normal distribution under \ref{ass:eps}-LRD with  $1/2<\gamma_1\le 1$, which follows from Theorem 4 of \citet{Hosking-1996}. The same Theorem implies  that $G(\cdot)$ is Normal under  \ref{ass:eps}-LRD with $\gamma_1=1/2$  when $a_t$ is normally distributed. Moreover,   $G(\cdot)$ is not Normal under \ref{ass:eps}-LRD with $0<\gamma_1<1/2$.
\end{remark}

A variant is to reduce the number of subsamples by introducing a random block selection with  $s(\cdot)$ estimated blockwise on subsamples of length $b$. Let $I_i$, $i=1,\ldots K$ be random variables indicating the initial point of every block of length $b$. We draw,  without replacement with uniform probabilities, the sequence $\set{I_i}_{i=1}^K$ from the set $I=\{1,2,\ldots,n-b+1\}$. The empirical distribution function of the subsampling variances of $\hat{\eps}_{t}$ over the random blocks is 
\begin{equation*}
\tilde{G}_{n,b}(x)=\frac{1}{K}\sum_{i=1}^{K} \uno \set{\tau_b \tonde{ \hat{V}_{n,b,I_i}-V_{n}} \leq x }.
\end{equation*}
In order to get the consistency of the subsample procedure both in the SRD and LRD cases, we consider two subsamples. The first one has a length of $b$ and we use it to estimate the signal, that is $s(\cdot)$. Instead, the second subsample, which is a subset of the first, has a length $b_1=o(b^{4/5})$ and we use this second subsample to estimate the variance and its distribution.  
The following result states the consistency of $\tilde{G}$ in approximating $G$.
\begin{theorem}\label{th:conv_var_dist}
Assume \ref{ass:s}, \ref{ass:eps}, \ref{ass:K} and \ref{ass:M}. Suppose that $\{a_t\}$, in \ref{ass:eps}, is Normally distributed when $0<\gamma_1\le 1/2$. Let $\hat s(t)$ be the estimate of $s(t)$ on a subsample of length $b$. Let $n\to \Inf$, $b\to \infty$, $b/n \to 0$, $b_1=o(b^{4/5})$ and  $K \to \Inf$, then $\sup_x\left|\tilde{G}_{n,b_1}(x)-G(x)\right| \convp 0$.
\end{theorem}
Proof of Theorem~\ref{th:conv_var_dist} is given in the Appendix. In analogy with what we have seen before we also establish consistency for the quantiles of $\tilde{G}(\cdot)$. Let $\tilde{q}_{n,b_1}(\gamma_2)$ be the $\gamma_2$-quantile with respect to  $\tilde{G}(\cdot)$.
\begin{corollary}\label{cor:2} %
Assume \ref{ass:s}, \ref{ass:eps}, \ref{ass:K} and \ref{ass:M}. Suppose that $\{a_t\}$, in \ref{ass:eps}, is Normally distributed when $0<\gamma_1\le 1/2$. Let $\hat s(t)$ be the estimate of $s(t)$ on a subsample of length $b$. Let $n\to \Inf$, $b\to \infty$, $b/n \to 0$, $b_1=o(b^{4/5})$ and  $K \to \Inf$, then $\tilde{q}_{n,b_1}(\gamma_2) \convp q(\gamma_2).$
\end{corollary}
Proof of Corollary \ref{cor:2} is given in the Appendix.

\begin{remark}\label{rem:free}  Note that the second subsample of length $b_1$ is a consequence of the optimal rate for the estimation of $s(t)$ subsample-wise.
\end{remark}

\REV{\begin{remark}\label{Vnhat} Following the same arguments as in the proof of Theorem \ref{th:kern_opt}, we have that $V_n-\sigma_{\varepsilon}^2=O_p(\tau_n^{-1})$ and 
\[
\hat V_n-\sigma_{\varepsilon}^2=\left\{
\begin{array}{ll}
O_p(\tau_n^{-1})& \mbox{ SRD and LRD with } 5/8\le\gamma_1\le 1\\
O_p\left( \Lambda_n^{-4/5}\right) & \mbox{ LRD with } 0<\gamma_1<5/8.\\
\end{array}
\right.
\]
Hence, $\tau_m\tau_n^{-1}\rightarrow 0$ and $\tau_m\Lambda_n^{-4/5}\rightarrow 0$, as $n\rightarrow\infty$, if $m=o(n^{4/5})$. Since $b_1=o(b^{4/5})$, we can argue that the consistency  of the subsampling procedure still holds if we replace $V_n$ with $\hat V_n$ in the previous results. 
\end{remark}
}

\REV{Remark \ref{Vnhat} and  Theorem \ref{th:conv_var_dist} ensure  that the statistical functional $\hat V_{n}$ provides the consistency for the subsample procedure. Note that in the case of SRD and LRD with $5/8\le \gamma_1\le 1$ we would also have  the consistency if $b_1=b$. 
}

Now, by using the previous results, we can state that the subsample strategy is consistent to estimate the asymptotic distribution of  $\tau_m(\widehat{SNR}-SNR)$ where $\widehat{SNR}$ is defined in (\ref{eq:Qn}).  The statistic $\widehat{SNR}$ has the numerator and denominator depending on $n$ and $m$, respectively. The latter is  mimed  in the  the subsample procedure.  In fact,  a subsample of length $b$ is used for the estimation of the signal power,  while a  subsample of length $b_1$ is used to estimate the variance of the error term. 

\begin{theorem}\label{th:convergence_dist_Q}
Let $\mathcal{Y}_n:=\{ y_1, y_2, \ldots, y_n\}$ be  a sampling realization of $\ProcY$. Assume \ref{ass:s}, \ref{ass:eps}, \ref{ass:K} and \ref{ass:M}. Suppose that $\{a_t\}$, in \ref{ass:eps}, is Normally distributed when $0<\gamma_1\le 1/2$. Assume 
$n\to \Inf$, $b\to \infty$, $b/n \to 0$, $b_1=o(b^{2/5})$, $m=o(n^{2/5})$, $b_1/m\rightarrow 0$ and  $K \to \Inf$, then
\[
\sup_x|\mathbb{Q}_n(x)-\mathbb{Q}(x)|\stackrel{p}{\longrightarrow} 0
\] 
where
\[
\mathbb{Q}_n(x):= \frac{1}{K}\sum_{i=1}^K\mathbb{I}\left\{\tau_{b_1}(\widehat{SNR}_{n,b,I_i}-\widehat{SNR})\le x\right\}, \mbox{ with } \widehat{SNR}_{n,b,I_i}:=10\log_{10}\left(\frac{\hat{U}_{n,b,I_i}}{\hat{V}_{n,b_1,I_i}}\right)
\]
and $\mathbb{Q}(x)$ is the asymptotic distribution of $\tau_m(\widehat{SNR}-SNR)$.
\end{theorem}
Proof of Theorem \ref{th:convergence_dist_Q} is given in the appendix. 
\REV{Note that in Theorem \ref{th:convergence_dist_Q} we need $b_1=o(b^{2/5})$ instead of $b_1=o(n^{4/5})$ found in previous results.  The reason for this is that the statistical functional $\widehat{SNR}$ is more complex than  $\hat V_n$ and a different relative speed for the  secondary block size $b_1$ is required.
}
Theorem \ref{th:convergence_dist_Q} provides the theoretical justification for the consistency of the subsample procedure with respect to the statistic $\widehat{SNR}$. Let $q^{Q}(\gamma_2)$ and $\tilde q^{Q}_{n,b_1}(\gamma_2)\equiv \tilde q^{Q}_{n,b_1}(\gamma_2|\tau_{b_1})$ be the quantiles with respect to $\mathbb{Q}(x)$ and $\mathbb{Q}_n(x)$, respectively. Note that we write $\tilde q^{Q}_{n,b_1}(\gamma_2|\tau_{b_1})$ to highlight the dependence from the scaling factor $\tau_{b_1}$ as in Section 8 of \cite{politis-etal-1999b}.  The main goal is to do  inference for $SNR$ without estimating the long memory parameter, and without using the sample statistic $\widehat{SNR}$. In this way, we do not need to fix or estimate $m$. To do this, we use Lemma 8.2.1 in \cite{politis-etal-1999b}. First, $\mathbb{Q}(x)$ always has a \REV{strictly positive} density function, at least, in a subset of real line (see \cite{Hosking-1996} and references therein). So, by Lemma 8.2.1 in \cite{politis-etal-1999b}, and using the same arguments as in the proof of Corollary \ref{cor:2}, we have that
\begin{equation}
\label{quan01}
\tilde q^{Q}_{n,b_1}(\gamma_2|\tau_{b_1})=q^{Q}(\gamma_2)+o_p(1).
\end{equation}
Following the same lines as in Section 8 of  \cite{politis-etal-1999b}, we have that
\begin{equation}
\label{quan02}
\tilde q^{Q}_{n,b_1}(\gamma_2|1)=\frac{\tilde q^{Q}_{n,b_1}(\gamma_2|\tau_{b_1})}{\tau_{b_1}}+\widehat{SNR}.
\end{equation}
Note that $\tilde q^{Q}_{n,b_1}(\gamma_2|1)$ is the quantile with respect to the empirical distribution function $1/K\sum_{i=1}^K\mathbb{I}\left(\widehat{SNR}_{n,b,I_i}\le \frac{x}{\tau_{b_1}}+\widehat{SNR}\right)$. Therefore, by (\ref{quan01}) and (\ref{quan02}) it follows that
\[
\tilde q^{Q}_{n,b_1}(\gamma_2|1)\tau_{b_1}=\tau_{b_1}\widehat{SNR}+q^{Q}(\gamma_2)+o_p(1).
\]
Since $\widehat{SNR}=SNR+O_p(\tau_m^{-1})$ and $\tau_{b_1}/\tau_{m}\rightarrow 0$ when $n\rightarrow\infty$, we have that
\[
\tilde q^{Q}_{n,b_1}(\gamma_2|1)=SNR+\frac{q^{Q}(\gamma_2)}{\tau_{b_1}}+o_p(\tau_{b_1}^{-1}).
\]
Therefore, a confidence interval for $SNR$ with a nominal level of $\gamma_2$ is given by  
\begin{equation}\label{eq:SNRquantiles}
\left[\tilde q^{Q}_{n,b_1}(\gamma_2/2|1),\quad \tilde q^{Q}_{n,b_1}(1-\gamma_2/2|1)\right].
\end{equation}
It is possible to consider the methods of self-normalization as in \cite{Jach-etal-2012}, and the estimation of  the scaling factor $\tau$ as in \cite{politis-etal-1999b}. These methods  would lead to more efficient confidence bands, in the sense that these would be  first order correct with a rate of $\tau_m^{-1}$ instead of $\tau_{b_1}^{-1}$. However, this would require the estimation of the unknown constants as in \cite{Jach-etal-2012}.

\begin{remark}\label{rem:}
In Theorems \ref{th:conv_var_dist} and \ref{th:convergence_dist_Q} the definition of $\tilde{G}(\cdot)$ and $\mathbb{Q}_n(\cdot)$ depend on statistics ($V_n$ and $\widehat{\text{SNR}}$) computed on the whole sample. On the other hand these two theorems give the theoretical framework for computing confidence interval as in \eqref{eq:SNRquantiles}, and this calculations will not require any calculation on the entire sample. %
\REV{%
In other words, $V_n$, $\hat V_n$, and $\widehat{SNR}$ are needed to center the involved  distributions, but not needed to approximate the quantiles as in  \eqref{eq:SNRquantiles}. Therefore, in this work  these quantities only have a theoretical role to show that the subsample procedure does not produce  degenerate asymptotic distributions.%
}
\end{remark}

\section{Numerical experiments}\label{sec:numeric}
In this section we present numerical experiments on simulated data. The assumptions given in this paper are rather general, and it is not possible to design a computer experiment that can be considered representative of all the kind of structures consistent with \ref{ass:s}--\ref{ass:M}. Here we assess the performance of Algorithm \ref{algo} under different scenarios for the structure of the noise term.  In order to do this we keep the structure of true signal fixed, and we investigate three variations of the noise data generating process. Data are sampled at fixed sampling frequency set at Fs = 44100Hz, a common value  in audio applications. Let $[0,T]$ be the data acquisition time interval, where $T$ is the duration of  the simulated signal in seconds. The signal is sampled at time $t=t_1, t_2, \ldots, T$, with $t_i = (i-1) / \text{Fs}$ for $i=1,2,\ldots, T{\times}\text{Fs}$, as follows
\[
y_{i} =  A_s \, \sin(2 \pi \,  50 \, t_i)   + \eps_i,   \qquad \text{with}\quad %
{i=1,2,\ldots T{\times}\text{Fs}},
\]
and $t_i = i/({T{\times}\text{Fs}})$. Therefore, the signal consists of a sinusoidal wave that produces energy at 50Hz. The signal power is equal to $A_s^2/2$, where $A_s$ is a scaling constant properly tuned to achieve a given true SNR. We set $T=30\text{sec}$ (implying $n=1,323,000$), and we consider the following  three cases for the noise. 
\begin{description}

\item[AR] The noise is generated from an AR(1)  process with independent normal innovations. This produces serial correlation in the error term and represents a case for SRD. In particular $\eps_i = -0.7\eps_{i-1} + u_i$,  where $\{u_i\}$ is an i.i.d. sequence with distribution $\text{Normal}(0, A_\eps)$, where $A_{\eps}$ is set to achieve a certain SNR. 

\item[P1] The random sequence $\{\eps_i \}$ has power spectrum equal to  $P(f) = A_{\eps} / f^{\beta}$, where $P(f)$ is the power spectral density at frequency $f$Hz. Here $\beta=0.2$ which induces some moderate LRD in $\{\eps_i\}$. The scaling constant $A_\eps$ is set to achieve the desired SNR.

\item[P2] same as \textbf{P1} but with $\beta=0.6$. This design introduces a much stronger LRD.
\end{description}
In \textbf{P1} and \textbf{P2} the noise has a so-called $1/f^\beta$-``power law'' where $\beta$ controls the amount of long range dependence. Larger values $\beta$ implies slower rate of decays for the serial correlations.  For $\beta=1$ pink noise is obtained. Values of $\beta \in [0,1]$ give a behavior between the white noise and the pink noise. In the case \textbf{P1}, $\gamma_1=1-\beta=0.8$ in \ref{ass:eps}-LRD. So, the asymptotic distribution of $\tau_n(V_n-\sigma_{\varepsilon}^2)$ is Normal. Whereas, in the case \textbf{P2}, $\gamma_1=1-\beta=0.4$. This implies that  the asymptotic distribution of $\tau_n(V_n-\sigma_{\varepsilon}^2)$ is not Normal (see Remark \ref{rem:distribuzione}). 

\textbf{P1} and \textbf{P2} are simulated based on the  algorithm of  \cite{TimmerKoenig1995}  implemented in the \textsf{tuneR} software of \cite{CRAN_tuneR}.  For each of the three sampling designs we consider two values for the true SNR: 10dB and 6dB. In most applications an SNR = 6dB is considered a rather noisy situations. We recall that at 6dB  the signal power is circa only  four times the variance of the noise, and 10dB means that the signal power is ten times the noise variance. There is a challenging aspect of these designs. The case with  \textbf{P2} noise and SNR = 6dB, is particularly difficult for our method. In fact, \textbf{P2} puts relatively large amount of variance (power) at low frequencies around 50Hz, so that the signal is not well distinguished from some spectral components of the noise. 
The two parameters of Algorithm  \ref{algo} are $b$ and $K$. We consider three  settings for the subsample window: $b=10\text{ms}=441$ samples,   $b=15\text{ms}=662$ samples, %
\REV{%
and $b$  estimated based on the method proposed in \cite{GoetzeRackauskas2001}. In the latter  case the optimal $b$ is computed over a grid ranging from $b=2$ms to $b=20$ms. 
In many applications is not easy to fix a value for $b$. However, in certain situations researchers have an idea about  the structure of the signal, and the time series is windowed with blocks of a certain length. In applications where  the  underlying signal is expected to be composed by harmonic components, the usual practice is to take blocks of size approximately equal to the  period of the harmonic component with the lowest expected frequency. The rational is to take the smallest window size so that each block is still expected to carry  some information about the low frequency  components.
For example for speech data usually blocks of 10ms are normally considered \citep{HaykinKosko2001}, whereas for music data  50ms is a common choice \citep{WeihsEtal2016}.   Note that the artificial data here have an harmonic component at 50Hz with a period of 20ms, and  we consider the fixed alternatives $b=10$ms and $b=15$ms as a robustness check. 
}%
We set $K=200$, of course larger values of $K$ would ensure less subsample induced variability. The $b_1$, i.e. the window length of the secondary subsample needed to estimate the distribution of the sampling variance, is set according to Theorem~\ref{th:conv_var_dist}. This is achieved by setting $b_1 = [b^{2/5}]$. For each combination of noise type, SNR, and $b$ we considered 500 Monte Carlo replica and we computed statistics to assess the performance of the procedure. Two aspects of the method are investigated corresponding to the two main contributions of the paper.

\begin{table}[!h]
\centering
\caption{%
Monte Carlo averages for the Mean Square Error (MSE) of the estimated signal power. Standard errors for the Monte Carlo averages are given in parenthesis. The block length $b$ is expressed in milliseconds [ms]. The case when $b=$opt corresponds to the estimated optimal block length reported in Table \ref{tab:mc_opt_b}.
}%
\label{tab:mse_power}
\begin{tabular}{cccc}
  \toprule
  Error Model       & True SNR & $b$ [ms] & MSE              \\
  \midrule
  AR                & 6        & 10       & 0.0016(0.000010) \\ 
  AR                & 6        & 15       & 0.0008(0.000005) \\ 
  AR                & 6        & opt      & 0.0016(0.000030) \\ 
  AR                & 10       & 10       & 0.0046(0.000024) \\ 
  AR                & 10       & 15       & 0.0022(0.000012) \\ 
  AR                & 10       & opt      & 0.0049(0.000074) \\
  \midrule
  P1                & 6        & 10       & 0.0018(0.000011) \\ 
  P1                & 6        & 15       & 0.0012(0.000008) \\ 
  P1                & 6        & opt      & 0.0020(0.000024) \\ 
  P1                & 10       & 10       & 0.0051(0.000030) \\ 
  P1                & 10       & 15       & 0.0037(0.000024) \\ 
  P1                & 10       & opt      & 0.0049(0.000058) \\
  \midrule
  P2                & 6        & 10       & 0.0024(0.000018) \\ 
  P2                & 6        & 15       & 0.0020(0.000015) \\ 
  P2                & 6        & opt      & 0.0025(0.000027) \\ 
  P2                & 10       & 10       & 0.0060(0.000038) \\ 
  P2                & 10       & 15       & 0.0051(0.000036) \\ 
  P2                & 10       & opt      & 0.0060(0.000056) \\
  \bottomrule
\end{tabular}
\end{table}

\begin{table}[!h]
\centering
\caption{Monte Carlo averages for the optimal $b$ estimated using the method proposed in \cite{GoetzeRackauskas2001}. Standard errors for the Monte Carlo averages are given in parentesis. The block length $b$ is expressed in milliseconds [ms].}
\label{tab:mc_opt_b}
\begin{tabular}{ccr}
  \toprule
  Error & True & Optimal $b$ [ms]       \\
  Model & SNR  & {}                     \\
  \midrule
  AR          & 6    & 9.44(0.193)      \\ 
  AR          & 10   & 8.93(0.164)      \\
  \midrule
  P1          & 6    & 8.18(0.168)      \\ 
  P1          & 10   & 11.05(0.162)     \\
  \midrule
  P2          & 6    & 9.36(0.198)      \\ 
  P2          & 10   & 10.80(0.172)     \\ 
	\bottomrule
\end{tabular}
\end{table}

\begin{table}[!h]
  \centering
  \caption{Monte Carlo averages for the absolute deviation of the estimated quantiles of the SNR distribution from the true counterpart. SNR is expressed in decibels. The true quantiles of the SNR distribution are computed based on Monte Carlo integration. Standard errors for the Monte Carlo averages are given in parenthesis.  The block length $b$ is expressed in milliseconds [ms]. The case when $b=$opt corresponds to the estimated optimal block length reported  Table \ref{tab:mc_opt_b}.}
  \label{tab:qmae}
  \begin{tabular}{cccccccc}
    \toprule
    Error & True & $b$  & \multicolumn{5}{c}{Quantile level}                                  \\
    \cmidrule{4-8}
    Model & SNR  & [ms] & 0.1         & 0.25        & 0.5         & 0.75        & 0.90        \\
    \midrule
    AR    & 6    & 10   & 4.99(0.013) & 2.99(0.022) & 0.28(0.010) & 1.50(0.008) & 2.09(0.007) \\ 
    AR    & 6    & 15   & 1.99(0.008) & 1.30(0.011) & 0.22(0.007) & 0.80(0.006) & 1.13(0.006) \\ 
    AR    & 6    & opt  & 4.29(0.125) & 2.19(0.059) & 0.28(0.010) & 0.94(0.017) & 1.38(0.025) \\ 
    AR    & 10   & 10   & 5.00(0.012) & 3.04(0.021) & 0.28(0.010) & 1.49(0.008) & 2.08(0.007) \\ 
    AR    & 10   & 15   & 2.01(0.007) & 1.30(0.011) & 0.19(0.007) & 0.81(0.006) & 1.13(0.006) \\ 
    AR    & 10   & opt  & 4.70(0.113) & 2.43(0.053) & 0.30(0.011) & 1.01(0.014) & 1.45(0.021) \\ 
    \midrule
    P1    & 6    & 10   & 5.37(0.019) & 3.09(0.024) & 0.34(0.011) & 1.87(0.007) & 2.52(0.006) \\ 
    P1    & 6    & 15   & 2.41(0.015) & 1.26(0.013) & 0.37(0.010) & 1.18(0.007) & 1.68(0.007) \\ 
    P1    & 6    & opt  & 4.58(0.094) & 2.13(0.047) & 0.39(0.013) & 1.27(0.015) & 1.76(0.019) \\ 
    P1    & 10   & 10   & 5.60(0.022) & 3.07(0.026) & 0.39(0.012) & 1.93(0.009) & 2.56(0.007) \\ 
    P1    & 10   & 15   & 2.77(0.025) & 1.32(0.015) & 0.49(0.013) & 1.35(0.007) & 1.86(0.008) \\ 
    P1    & 10   & opt  & 4.46(0.087) & 2.20(0.048) & 0.45(0.012) & 1.54(0.012) & 2.05(0.013) \\ 
    \midrule
    P2    & 6    & 10   & 5.19(0.031) & 2.36(0.026) & 0.69(0.016) & 2.52(0.011) & 3.58(0.012) \\ 
    P2    & 6    & 15   & 2.83(0.028) & 0.83(0.019) & 0.99(0.015) & 2.17(0.012) & 3.05(0.013) \\ 
    P2    & 6    & opt  & 4.00(0.082) & 1.41(0.043) & 0.81(0.017) & 2.02(0.018) & 2.89(0.020) \\ 
    P2    & 10   & 10   & 5.16(0.035) & 2.02(0.031) & 1.03(0.017) & 2.65(0.011) & 3.56(0.012) \\ 
    P2    & 10   & 15   & 2.82(0.038) & 0.64(0.018) & 1.39(0.016) & 2.37(0.011) & 3.07(0.011) \\ 
    P2    & 10   & opt  & 4.29(0.082) & 1.46(0.049) & 1.12(0.020) & 2.32(0.014) & 3.05(0.015) \\ 
    \bottomrule
  \end{tabular}
\end{table}

The first contribution of the paper is  Theorem~\ref{th:kern_opt}, where optimality and consistency of the Priestley-Chao kernel estimator is established under rather general assumptions on the error term. The kernel smoothing is used in Algorithm \ref{algo} to estimate the signal power in the numerator of \eqref{eq_algo_snr}. In Table \ref{tab:mse_power} we report the Monte Carlo averages for the Mean Square Error (MSE) of the estimated signal power. Going from the simplest \textbf{AR} to the complex \textbf{P2} noise model there is an increase in MSE as expected. The longer $b=15$msec subsample window always produced better results. The apparently counterintuitive evidence is that for larger amount of noise (lower SNR), the signal power is slightly better estimated. In order to understand this, note that the noise (in all three cases) produces most of its power in a low frequency region containing the signal frequency (i.e. 50Hz). In the lower  noise case there is still a considerable amount of noise acting at low frequency that the adaptive nature of the kernel smoother is not able to recognize properly. %
\REV{%
In Table \ref{tab:mc_opt_b} we report the estimated average $b$ with its Monte Carlo standard error. The estimated  $b$ is always near 10ms, and the latter  produced results that are only  slightly worse  than  those obtained for fixed $b=15$ms. 
}%

The second contribution of the paper is the consistency result (see Theorem~\ref{th:convergence_dist_Q} and related results) for the distribution of the SNR statistic.  In order to measure the quality of method one needs to define the ground truth in terms of the sampling distribution of the target SNR statistic. The derivation of an expression for such a distribution would be an  analytically intractable. Therefore, we computed the quantiles of the true SNR statistic based on Monte Carlo integration, and in Table \ref{tab:qmae} we report the average absolute differences between estimated quantiles and the true counterpart. Based on Corollary \ref{cor:2} the convergence of the distribution of the SNR is mapped into its quantiles, therefore this makes sense. Comparison involves five different quantile levels to assess the behavior of the procedure both in the tails and in the center of the distribution. The average deviations of Table \ref{tab:qmae} are computed in decibels. Overall the method can capture the center of the distribution pretty well in all cases. The estimation error increases in the tails of the distribution as one would expect. The right tail is estimated better than the left tail. In all cases  the performance in the tails of the SNR distribution is better captured with a $b=15$msec window, although in the center of the distribution the differences implied by different values of  $b$ are much smaller. Going from SNR = 6 to SNR = 10 results are clearly better on the left tail of the distribution especially in the case \textbf{P2}. Again the estimated version of  $b$ pushes the corresponding results towards the $b=10$ms case. 

Every method has its own tunings, and the evidence here is that $b$ has some effects on the proposed method. The major impact of $b$ is about the tails of the SNR distribution. \REV{The selection of $b$ based on the method proposed by \cite{GoetzeRackauskas2001} deliver a fully satisfying solution that does not require any prior knowledge on the data  structure. The only drawback of estimating $b$ is that the overall algorithm needs to be executed for several candidate values of $b$.}
As final remark we want to stress that the method proposed here is designed to cope with much larger values of $n$. In this experiment the sampling is repeated a number of  times to produce Monte Carlo estimates, therefore we had to choose an $n$ compatible with reasonable computing times according to the available hardware. A limited number of trials with $T$ up to several minutes  (which implies that $n$ goes up to several millions) have been  successfully tested without changing the final results. Therefore, we can conclude that the algorithm scales well with the sample size.


\section{Application to EEG data}\label{sec:application}

In this section we illustrate an application of the proposed methodology to electroencephalography (EEG) data obtained from the PhysioNet repository \citep{GoldbergerEtal2000}. In particular we considered the ``CHB-MIT Scalp EEG Database'' available at  \url{https://physionet.org/pn6/chbmit/}. The database contains EEG traces recorded at the Children’s Hospital Boston on  pediatric subjects with intractable seizures. Subjects were monitored for several days after the  withdrawal of anti-seizure medication before the final decision about the  surgical intervention. 22 subjects were traced during the experiment for several days using the  international 10-20 EEG system. The latter is a standard that specifies  electrode positions and nomenclature. Therefore, for each subject 21 electrodes have been placed in certain positions of the scalp, each of these electrodes produced an electric signal sampled at 256Hz and measured with 16bit precision. This means that each day (24 hours), the EEG machine produced 21 time series each containing  $n = 22,118,400$ data points for a total of 464,486,400 amplitude measurements for each subject in the experiment. A description of the ``\textit{CHB-MIT Scalp EEG Database}'', as well as details about the data acquisition is given in \cite{Shoeb2009}.

\begin{figure}[!t]
  \centering
  \includegraphics[width=\linewidth]{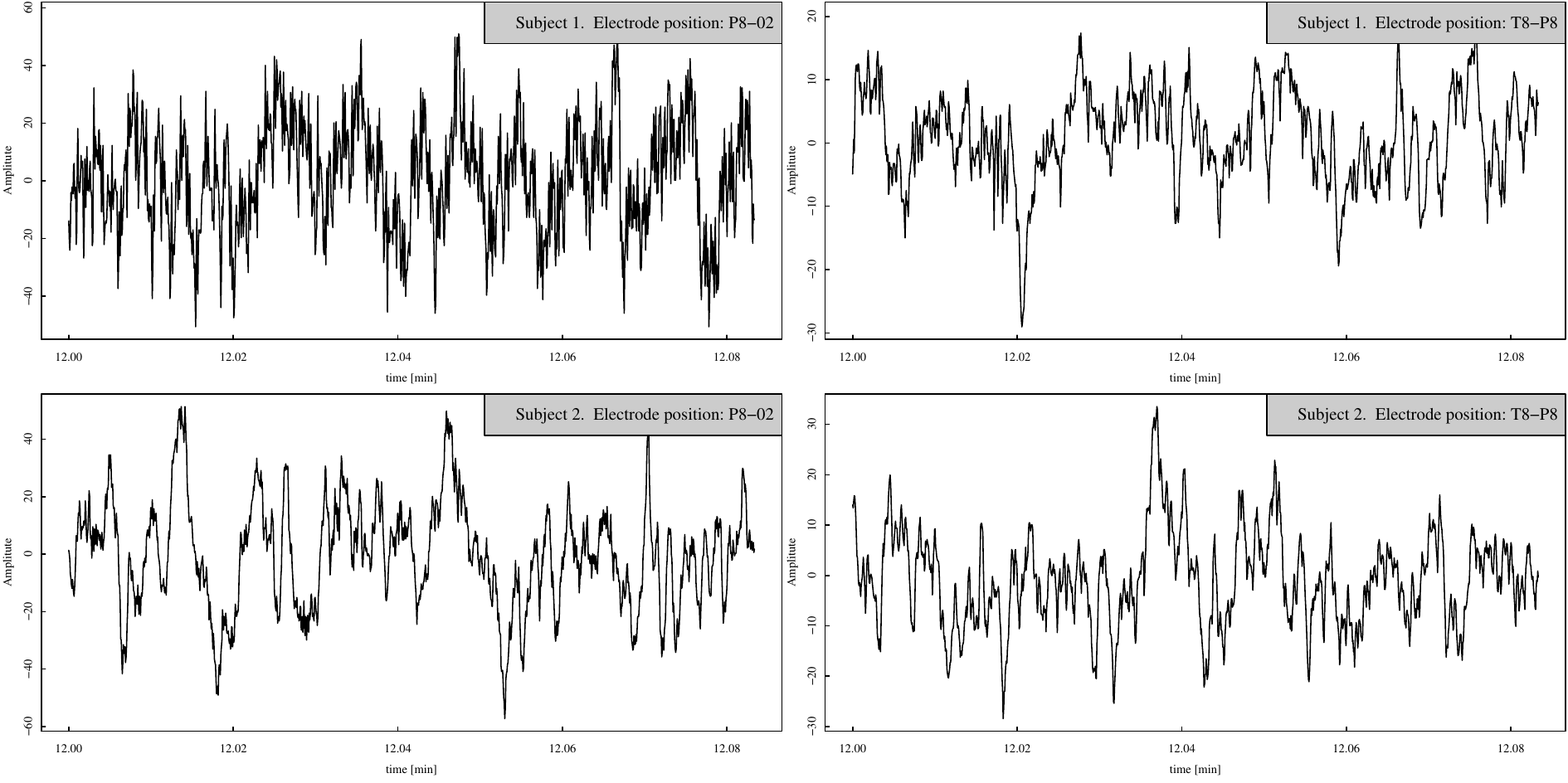}
  \caption{Time series plots of the amplitude of the EEG signals  recorded in positions P8-02 and T8-P8 for two distinct subjects in the experiment. The duration of the fragment is 5sec, and it starts at 12min from the beginning of the recording.}\label{fig:eeg}
\end{figure}

EEG signals have complex structures. Various sources of noise can be injected in the measurement chain, therefore it is always of interest to understand the behavior of the  SNR.  For this application we considered data for the first 3 subjects of the database, and we considered two electrode positions labeled P8-02 and T8-P8 in the  10-20 EEG system. The P8-02 electrode is placed on the parietal lobe responsible for integrating sensory information of various types.  The T8-P8  electrode is placed on the  temporal lobe which  transforms  sensory inputs  into  meanings retained as  visual memory, language comprehension, and emotion association. An example of these traces is given in Figure~\ref{fig:eeg}.

The method proposed here has been applied to obtain confidence intervals for the SNR. An SNR$\geq 10$dB can be considered a requirement for a favorable noise floor in these applications.  In order to assess the robustness of the procedure with respect to the choice of the subsampling window $b$, for each case we considered windows of fixed size $ b = \{3\text{sec},5\text{sec},7\text{sec}\}$ which means $b=\{768, 1280, 1792\}$, %
\REV{%
plus the estimated $b$ with the method proposed by \cite{GoetzeRackauskas2001}. 
The estimation of $b$ is performed on a grid of equispaced points between 2sec and 10sec.} %
The literature about EEG signals doesn't tell us whether the processes involved have a clear time scale, but 5sec is considered approximately the time length needed to identify interesting cerebral activities. For each $b$ the corresponding $b_1$ is set according to  $b_1 = [b^{2/5}]$ as for the numerical experiments. 
\REV{%
In Table~\ref{tab:confidence_interval} lower and upper limits for 90\% and 95\% confidence intervals of the SNR are reported. %
 Overall the results with estimated $b$ are comparable  to those with fixed $b$. %
 The upper limits of these confidence intervals is never smaller than 10dB. The lower limits is negative in all cases, which means that for all cases there is a chance that the power of the stochastic component dominates that of the deterministic component in model  \eqref{eq:Yt}.  While the upper limit of these confidence intervals is rather stable across units for the same $b$ value, larger differences are observed in terms of the lower limit. All this is a clear indication of the asymmetry of the SNR statistic. But this is expected since the two tails of the SNR statistic reflects two distinct mechanisms. In fact, a negative value for the SNR statistic (left tail) corresponds to situations where the dynamic of the observed time series is driven by the error term of equation \eqref{eq:Yt}. On the other hand a positive value of the SNR statistic (right tail) corresponds to situations where the dynamic is driven by the smooth changes induces by $s(\cdot)$. Ceteris paribus, going from 90\%  level to 95\% does not change the results dramatically. Note that in this kind of applications a 3dB difference is not considered a large difference. Regarding data recorded in the P8-02 position the length of the confidence interval, going from 90\% to 95\% changes between 1.28dB to 3.6dB, where the maximum variation is measured for Subject 3 when $b=7$sec.  For the T8-P8 case  the length of the confidence interval, going from 90\% to 95\%, changes between 1.52dB to 3.83dB, and here the  maximum variation is measured for Subject 2 when $b=7$sec. Some pattern is observed across experimental units.  For given confidence level and  $b$, overall Subject 1  reports the shortest confidence intervals.  Subject 2 reports the longest  intervals for records in  position P8-02. Subject 3 reports the longest intervals in position T8-P8. The variations across values of $b$, with all else equal, are not dramatic. The settings with $b=3$, and $b=5$,  produced longer intervals if compared with $b=7$sec and the optimal $b$. The data-driven method of \cite{GoetzeRackauskas2001} produced an estimated $b$ in the range [3sec, 7sec] for P8-02 data, and  [6sec, 8sec] for T8-P8 data. These values are comparable with the rule of thumb that 5sec is a reasonable time scale for the kind of signals involved here. The general conclusion is that in absence of relevant information, the method of \cite{GoetzeRackauskas2001} gives a useful data-driven choice of $b$. 
}%

\begin{table}[!h]
\centering
\caption{%
Lower and upper limit of the confidence interval for the SNR (expressed in dB) of the EEG records in positions P8-02 and T8-P8. The block length $b$ is expressed in seconds.
The notation $b=$opt denotes the case when $b$ is estimated.%
}%
\label{tab:confidence_interval}
\begin{tabular}{ccccccc}
  	\toprule
  	         Confidence Level           & Subject & $b$~[sec] & \multicolumn{2}{c}{P8-02} & \multicolumn{2}{c}{T8-P8} \\
  	\cmidrule(l){4-5} \cmidrule(l){6-7} &         &           & Lower                     & Upper & Lower & Upper     \\
  \midrule
  90\%                                      & 1       & 3         & -0.93                     & 11.18 & -2.85 & 11.92     \\ 
                                            &         & 5         & 1.14                      & 11.26 & -3.41 & 11.92     \\ 
                                            &         & 7         & 0.86                      & 12.25 & -3.04 & 12.76     \\ 
                                            &         & opt       & -2.24                     & 12.52 & -2.56 & 14.70     \\
 \cmidrule(l){2-7} 
                                            & 2       & 3         & -9.04                     & 14.93 & -8.38 & 14.00     \\ 
                                            &         & 5         & -9.67                     & 15.14 & -7.88 & 12.97     \\ 
                                            &         & 7         & -8.29                     & 15.45 & -7.41 & 12.53     \\ 
                                            &         & opt       & -8.91                     & 15.68 & -7.57 & 12.42     \\
 \cmidrule(l){2-7} 
                                            & 3       & 3         & -8.30                     & 13.06 & -7.14 & 16.50     \\ 
                                            &         & 5         & -8.30                     & 13.41 & -5.37 & 16.97     \\ 
                                            &         & 7         & -5.81                     & 13.01 & -5.01 & 15.38     \\ 
                                            &         & opt       & -7.24                     & 13.49 & -6.70 & 15.66     \\
  \midrule
  95\%                                      & 1       & 3         & -1.36                     & 12.19 & -3.81 & 13.19     \\ 
                                            &         & 5         & 0.50                      & 13.00 & -4.30 & 12.93     \\ 
                                            &         & 7         & 0.49                      & 15.45 & -3.63 & 14.81     \\ 
                                            &         & opt       & -3.11                     & 14.36 & -3.48 & 15.30     \\
  \cmidrule(l){2-7}
                                            & 2       & 3         & -9.65                     & 15.60 & -9.54 & 15.02     \\ 
                                            &         & 5         & -10.47                    & 16.46 & -8.98 & 14.02     \\ 
                                            &         & 7         & -8.96                     & 16.17 & -8.96 & 14.81     \\ 
                                            &         & opt       & -10.62                    & 16.81 & -8.16 & 13.49     \\
  \cmidrule(l){2-7}
                                            & 3       & 3         & -9.60                     & 15.04 & -8.09 & 18.22     \\ 
                                            &         & 5         & -10.66                    & 14.44 & -6.80 & 18.68     \\ 
                                            &         & 7         & -7.46                     & 14.96 & -6.58 & 16.07     \\ 
                                            &         & opt       & -9.75                     & 14.23 & -7.94 & 17.05     \\   
  \bottomrule
  \end{tabular}
\end{table}

\section{Conclusions and final remarks}\label{sec:conclusions}
In this paper we developed an estimation method that consistently estimates the distribution of a SNR statistic in the context of time series data with errors belonging to a rich class of stochastic processes.  We restricted the model to the case where the signal is a smooth function of time. The theory developed here can be easily adapted to more general time series additive regression models. The reference model for the observed data, and the theory developed here adapts to many possible applications that will be the object of a distinct paper.   In this work we concentrated on the theoretical guarantees of the proposed method. The estimation is based on a random  subsampling algorithm that can cope with massive sample sizes. Both the smoothing,  and the subsampling techniques at the earth of Algorithm \ref{algo} embodies original innovations compared to the existing literature on the subject. Numerical experiments described in Section  \ref{sec:numeric} showed that the proposed algorithm performs well in finite samples.


\begin{center}
\textbf{\Large Appendix}
\end{center}

In this section we report the proofs of statements and some useful technical lemmas. First, we state a lemma to evaluate the
$\text{MISE}(\hat s;h)$.
\begin{lemma}\label{lemma1} 
Assume \ref{ass:s}, \ref{ass:eps} and \ref{ass:K}. For $t\in I_h=(h,1-h)$
\begin{equation*}
\text{MISE}(\hat s;h)=\frac{h^4R(s'')}{4}d_K+\sigma^2_{\varepsilon}\frac{N_K}{nh}+2\sigma^2_{\varepsilon}N_K\frac{S^*_{\rho}}{\Lambda_nh}+o\left(\frac{1}{\Lambda_nh}+h^4\right),
\end{equation*}
\[
\text{MISE}(\hat s;h)=\text{AMISE}(\hat s;h)+o\left(\frac{1}{\Lambda_nh}+h^4\right)
\]
where $\hat s$ is the kernel estimator in \eqref{eq:ker}, $R(s'')=\int_{I_h}[s''(t)]^2dt$, $d_K=\int u^2\mathcal{K}(u)du$, $N_K=\int \mathcal{K}^2(u)du$, $\sigma^2_{\varepsilon}=\ex[\varepsilon_t^2]$, $\Lambda_n$ is defined in \eqref{Lambda} and
\begin{equation}
\label{Srho} S^*_{\rho}:=
\begin{cases}
\lim_{n\rightarrow\infty}\sum_{j=1}^n\rho(j) & \text{if \ref{ass:eps}-SRD holds}\\
\lim_{n\rightarrow\infty}\frac{1}{\log n} \sum_{j=1}^n\rho(j) & \text{if \ref{ass:eps}-LRD holds with} \quad \gamma_1=1\\
\lim_{n\rightarrow\infty}\frac{1}{n^{1-\gamma_1}} \sum_{j=1}^n\rho(j)&  \text{if \ref{ass:eps}-LRD holds with} \quad 0<\gamma_1<1.
\end{cases}
\end{equation}
\end{lemma}
\begin{proof}
By  \ref{ass:K} it  follows that conditions A-C of \citet{altman-1990} are satisfied. Now, let
\begin{equation*}
\rho_n(j):=
\begin{cases}
\rho(j)                       & \text{if \ref{ass:eps}-SRD holds}\\ 
\frac{1}{\log n}\rho(j)       & \text{if \ref{ass:eps}-LRD holds with} \quad \gamma_1=1\\
\frac{1}{n^{1-\gamma_1}}\rho(j) & \text{if \ref{ass:eps}-LRD holds with} \quad 0<\gamma_1<1.
\end{cases}
\end{equation*}
For the cases SRD and LRD with $\gamma_1=1$ the conditions D and E of \citet{altman-1990} are still satisfied with $\rho_n(j)$. Following the same arguments as in the proof of Theorem 1 of \citet{altman-1990} the result follows. Finally, in the last case, $\rho_n(j)$ satisfies condition D but not condition E of \citet{altman-1990}. So, we have
\begin{equation*}
\sum_{j=1}^nj\rho_n(j)=O(n).
\end{equation*}
Therefore, using Lemma A.4 in \citet{altman-1990}, it follows that
\begin{equation*}
\var[\hat s]=\sigma^2_{\varepsilon}\frac{N_K}{nh} %
+ 2\sigma^2_{\varepsilon}N_K\frac{S^*_{\rho}}{\Lambda_nh} %
+ o\left(\frac{1}{\Lambda_nh}\right).
\end{equation*}
The latter completes the proof.
\end{proof}
The $\text{AMISE}(\hat s;h)$ is the asymptotic MISE, the main part of the MISE.

Note that Lemma \ref{lemma1} gives a similar formula to (2.8) in Theorem 2.1 of \citet{Hall-a-etal-1995}. However, differently  from \citet{Hall-a-etal-1995} our approach does not need to introduce an additional  parameter to capture SRD and LRD. Also notice that taking $h\in H$ as in \ref{ass:K}, implies that  $\text{MISE}(\hat s;h)=O\left(\Lambda_n^{-4/5}\right)$, which means that the kernel estimator achieves  the global optimal rate.

\paragraph{Proof of Theorem~\ref{th:kern_opt}.}  Lemma \ref{lemma1} holds under \ref{ass:s}, \ref{ass:eps} and \ref{ass:K}. Let %
$\hat{\gamma}(j)=\frac{1}{n}\sum_{t=1}^{n-j}\hat{\varepsilon}_t\hat{\varepsilon}_{t+j}$ %
be the estimator of the autocovariance $\gamma(j)$ with $j=0,1,\ldots$. By \ref{ass:K} $r_n=\frac{1}{\Lambda_nh}+h^4=\Lambda_n^{-4/5}$, and  by Markov inequality 
\begin{equation}\label{eq:proof_1_1}
P\left(\frac{1}{n}\sum_{i=1}^{n}\left(s(i/n)-\hat{s}(i/n)\right)^2>\eta\right)\le \frac{1}{\eta}\frac{1}{n}\sum_{i=1}^n E\left((\hat s(i/n) - s(i/n))^2\right)=O\left(\text{MISE}(\hat{s};h)\right),
\end{equation}
for some $\eta>0$ and when $n\rightarrow\infty$.
\\
It means that  $\frac{1}{n}\sum_{i=1}^{n}\left(s(i/n)-\hat{s}(i/n)\right)^2=\text{AMISE}(\hat{s};h)+o_p(r_n)$.
Rewrite  $\hat{\gamma}(j)$ as
\begin{align}
\hat{\gamma}(j) 
=&  \frac{1}{n}\sum_{i=1}^{n-j}\left(s(i/n)-\hat{s}(i/n)\right)\left(s((i+j)/n)-\hat{s}((i+j)/n)\right) + \nonumber \\  %
+&  \frac{1}{n}\sum_{i=1}^{n-j}\left(s((i+j)/n)-\hat{s}((i+j)/n)\right)\varepsilon_{i} +  \nonumber\\ %
+&  \frac{1}{n}\sum_{i=1}^{n-j}\left(s(i/n)-\hat{s}(i/n)\right)\varepsilon_{i+j}+\frac{1}{n}\sum_{i=1}^{n-j}\varepsilon_i\varepsilon_{i+j}=\text{I+II+III+IV}.
\label{eq_autocorr_decomp}
\end{align}
By \eqref{eq:proof_1_1}  and Cauchy-Schwartz inequality it results that term I$=O_p(r_n)$ in $\hat{\gamma}(j)$. Consider term III in \eqref{eq_autocorr_decomp}. 
\REV{Without loss of generality assume t $s(t)\not = 0$. By Chebyshev inequality
\[
P\left(\left|\hat s (t)-s(t)\right|>\eta\right)\le \frac{MSE(\hat s;h)}{\eta ^2},
\] 
for some $\eta>0$. By using the same arguments as in the proof of Lemma \ref{lemma1}, it follows that $MSE(\hat s;h)=O\left(r_n\right)$ so that   $\hat s(t)=s(t)(1+O_p(r_n^{1/2}))$.} Therefore, it is sufficient to investigate the behaviour of  
\begin{equation*}
\frac{1}{n}\sum_{i=1}^{n-j}s(i/n)\varepsilon_{i+j}.
\end{equation*}
$\sum_j^n\rho(j)=O(\log n)$ under LRD with $\gamma_1=1$, and   $\sum_j^n\rho(j)=O(n^{1-\gamma_1})$ under LRD with $0<\gamma_1<1$. By \ref{ass:s}, and  applying Chebyshev inequality, it happens that III$=O_p(\Lambda_n^{-1/2})$. Based on similar arguments one has that term II$=O_p(\Lambda_n^{-1/2})$. Now consider last term of \eqref{eq_autocorr_decomp}, and notice that it is the series of products of autocovariances. Theorem 3 in \citet{Hosking-1996} is used to conclude that the series is convergent under SRD and LRD with $1/2<\gamma_1\le 1$,  while it is divergent under LRD with $0<\gamma_1\le 1/2$. Based on this, direct application of Chebishev inequality to term IV implies that IV$=o_p(\Lambda_n^{-1/2}).$ Then $\hat{\gamma}(j)=\gamma(j) + O_p(r_n) +O_p(\Lambda_n^{-1/2}) +O_p(j/n)$, where the $O_p(j/n)$ is due to the bias of $\hat{\gamma}(j)$. This means that $\hat{\rho}(j) = \rho(j) + O_p(r_n)+ O_p(\Lambda_n^{-1/2})+O_p(j/n)$. Since $\mathcal{K}(\cdot)$ is bounded then one can write
\begin{equation*}
\frac{1}{nh}\sum_{j=-M}^M
\mathcal{K}\left(\frac{j}{nh}\right)\hat{\rho}(j)=\frac{1}{nh}\sum_{j=-M}^M
\mathcal{K}\left(\frac{j}{nh}\right)\rho(j)+\frac{M}{nh}O_p(r_n)+\frac{M}{nh}O_p(\Lambda_n^{-1/2})+\frac{M^2}{n}O_p\left(\frac{1}{nh}\right).
\end{equation*}
Using \ref{ass:M} and $h=O(\Lambda_n^{-1/5})$,  \ref{ass:K} implies that 
\begin{equation}\label{eq_spectral_consistency}
\frac{1}{nh}\sum_{j=-M}^M
\mathcal{K}\left(\frac{j}{nh}\right)\hat{\rho}(j)=\frac{1}{nh}\sum_{j=-M}^M
\mathcal{K}\left(\frac{j}{nh}\right)\rho(j)+o_p(r_n).
\end{equation}
Consider
\begin{equation*}
Q_1=\left|\frac{1}{nh}\sum_{j=-[nh/2]}^{[nh/2]}\mathcal{K}\left(\frac{j}{nh}\right)\rho(j)-\frac{1}{nh}\sum_{j=-M}^M \mathcal{K}\left(\frac{j}{nh}\right)\hat{\rho}(j)\right|,
\end{equation*}
and by \eqref{eq_spectral_consistency} it follows that
\begin{equation*}
Q_1=\left|\frac{2}{nh}\sum_{j=M+1}^{[nh/2]}\mathcal{K}\left(\frac{j}{nh}\right)\rho(j)\right|+o_p(r_n).
\end{equation*}
By \ref{ass:eps}, \ref{ass:K} and \ref{ass:M},
\begin{equation*}
\frac{1}{nh}\sum_{j=M+1}^{[nh/2]}\mathcal{K}\left(\frac{j}{nh}\right)\rho(j)\sim\rho(nh-M)=o(r_n),
\end{equation*}
which implies that $Q_1=o_p(r_n)$.  
It means that the CV function, as defined in (22) of \citet{altman-1990} with the estimated  correlation function, has an error rate of $o_p(r_n)$ with respect to 
\[
\left[1-\frac{1}{nh}\sum_{j=-M}^M\mathcal{K}\left(\frac{j}{nh}\right)\rho(j)\right]
^{-2}  \frac{1}{n} \sum_{i=1}^n \hat{\eps}_i^2.
\]
Now, we can apply the classical bias correction and based on (14) in \citet{altman-1990}, we
have that
\[
\text{CV}(h)=\frac{1}{n}\sum_{i=1}^n \eps_i^2  +\text{MISE}(\hat{s};h)+
o_p(r_n)=\sigma_{\varepsilon}^2+\text{AMISE}(\hat{s};h)+
o_p(r_n)
\]
Since $\text{AMISE}(\hat s;h)=O(r_n)$, it follows that $\hat{h}$, the
minimizer of $\text{CV}(h)$, is equal to $h^\star$, the minimizer
of $\text{MISE}(\hat{s};h)$, asymptotically in probability. By Lemma \ref{lemma1}, it follows that $h^\star$ is the same minimizer with respect to $\text{AMISE}(\hat s;h)$ asymptotically.
\qed\\

The subsequent Lemmas are needed to show Theorem~\ref{th:conv_var_dist} and Corollary \ref{cor:2}.
\begin{lemma}\label{lemma2} 
Assume  \ref{ass:eps}. Suppose that $\{a_t\}$, in \ref{ass:eps}, is Normally distributed when $0<\gamma_1\le 1/2$. Then $n\to\infty$, $b\to\infty$ and $b/n \to 0$ implies $\sup_x\left|G_{n,b}(x)-G(x)\right| \convp 0,$ and 
$q_{n,b}(\gamma_2) \convp q(\gamma_2)$ for all $\gamma_2 \in (0,1).$
\end{lemma}
\begin{proof}
Under \ref{ass:eps}-SRD, Theorems 4.1 and 5.1 of \citet{politis-etal-1999} hold and the results follow. The rest of the proof deals with the LRD case. Since $G(x)$ is continuous \citep[see][]{Hosking-1996}, we  follow  proof of Theorem 4 of \citet{Jach-etal-2012}. Fix 
 $G_{n,b}^0(x)=\frac{1}{N}\sum_{i=1}^N \uno\left\{\tau_b\left(V_{n,b,i}-\sigma_{\varepsilon}^2\right)\le x\right\}$ with $N=n-b+1$. It is sufficient to show that $\var[G_{n,b}^0(x)] \to 0$ as $n\to \infty$. Apply Theorem 2 in \citet{Hosking-1996} to conclude that $\tau_n\left(V_n-\sigma_{\varepsilon}^2\right)$ has the same distribution as $\tau_n\left(V_n^{1}\right)$ , where $V_n^{1}=\frac{1}{n}\sum_{i=1}^n(\varepsilon_i^2-\sigma_{\varepsilon}^2)$. Therefore, we have to show that  $\var[G^1_{n,b}(x)] \to 0$ as  $n \to \infty$,  where 
\begin{equation*}
G^1_{n,b}(x) = \frac{1}{N}\sum_{i=1}^N\uno\left\{\tau_bV_{n,b,i}^1\le x)\right\} %
\quad \text{with} \quad %
V_{n,b,i}^1=\frac{1}{b}\sum_{j=1}^b (\varepsilon_{j+i-1}^2-\sigma_{\varepsilon}^2).
\end{equation*}
Using the stationarity of $\ProcEps$, it follows that $\var[G_{n,b}^1(x)] = \ex[(G_{n,b}^1(x)-G_b^1(x))^2],$ where $G_b^1(x)=P\left(\tau_bV_b^1\le x\right)$. By \citet{Hall-b-etal-1998} the Hermite rank of the square function is 2. Then, based on the same arguments as in the proof of Theorem 2.2 of \citet{Hall-b-etal-1998} with $q=2$, we can write
\begin{equation}\label{Var:lemma2}
\var[G_{n,b}^1(x)] \le %
\frac{2b+1}{N}G_b^1(x)+\frac{2}{N}\sum_{i=b+1}^{N-1}\left|P\left(\tau_bV_{n,b,1}^1\le
x,\tau_bV_{n,b,i+1}^1\le x\right)-\left[G_b^1(x)\right]^2\right|.
\end{equation}
Consider
\begin{equation*}
\cov[\tau_bV_{n,b,1}^1,\tau_bV_{n,b,N}^1]=\frac{\tau_b^2}{b^2}\ex\left[\sum_{i=1}^b(\varepsilon_i^2-\sigma_{\varepsilon}^2)\cdot
\sum_{i=1}^b(\varepsilon_{i+N-1}^2-\sigma_{\varepsilon}^2)\right].
\end{equation*}
After some algebra, we obtain
\begin{equation}\label{Cov:lemma2}
\cov[\tau_bV_{n,b,1}^1,\tau_bV_{n,b,N}^1]=\frac{\tau_b^2}{b}\sum_{k=-(b-1)}^{b-1}\left(1-\frac{|k|}{b}\right)\phi_2(k+N),
\end{equation}
where for $k=1,2,\ldots,$ $\phi_2(k)$ are the autocovariances of
$\{\eps_t^2\}_{t_\in\Z}$.  For $k \to \infty$,  \ref{ass:eps}-LRD with $0<\gamma_1\le 1$ implies that  $\phi_2(k)=O(k^{-2\gamma_1})$  by Theorem 3 of \citet{Hosking-1996}. Take  \eqref{Cov:lemma2} and  note that
\begin{equation*}
\left|\cov[\tau_bV_{n,b,1}^1,\tau_bV_{n,b,N}^1]\right|\le
\frac{\tau_b^2}{b}\sum_{k=-(b-1)}^{b-1}|\phi_2(k+N)|\le
\frac{\tau_b^2}{b}\left(\frac{1}{N-b+1}\right)^{2\gamma_1}C_b,
\end{equation*}
where 
\begin{equation*}
C_b:= %
\begin{cases}
O(1) & 1/2<\gamma_1\le 1,\\
O(\log b) & \gamma_1=1/2,\\ 
O(b^{1-2\gamma_1}) & 0<\gamma_1<1/2.
\end{cases}
\end{equation*}
The latter implies that for $n\conv \infty$, \eqref{Cov:lemma2} converges to zero. Therefore, $\tau_bV_{n,b,1}^1$ and $\tau_bV_{n,b,N}^1$ are
asymptotically independent. The latter can be argued based on asymptotic normality when $1/2 \leq \gamma_1 \leq 1$. For the case $0 < \gamma_1 < 1/2$ the asymptotic independence can be obtained by using Theorem 2.3 of \cite{Hall-b-etal-1998}.  Thus, right hand side of \eqref{Var:lemma2} converges to zero as $n\conv \infty$ by Cesaro Theorem. The latter  shows that $\sup_x\left|G_{n,b}(x)-G(x)\right| \convp 0.$ \\
Following the same arguments as in Theorem 5.1 of \citet{politis-etal-1999}, and by using the  first part of this proof
one shows that  $q_{n,b}(\gamma_2)\stackrel{p}{\longrightarrow}q(\gamma_2)$. The latter completes the proof.
\end{proof}


\begin{lemma}\label{prop:2}
Assume \ref{ass:s}, \ref{ass:eps}, \ref{ass:K} and \ref{ass:M}. Suppose that $\{a_t\}$, in \ref{ass:eps}, is Normally distributed when $0<\gamma_1\le 1/2$.  Let $\hat s(t)$ be the estimate of $s(t)$ computed on the entire sample (of length $n$). Then $n \to \Inf$ and $b=o(n^{4/5})$ implies $\sup_x \abs{ \hat{G}_{n,b}(x)-G(x)} \convp 0.$
\end{lemma}
\begin{proof}
Denote  $r_n=\frac{1}{\Lambda_nh}+h^4$. By Lemma \ref{lemma1} and \ref{ass:K},  $r_n=\Lambda_n^{-\frac{4}{5}}$. $\hat s(t)$ is computed on the whole time
series. By Lemma \ref{lemma2}, we can use the same approach as in Lemma 1, part (\textit{i}) of \cite{CorettoGiordano2016x}. We have only to verify that $\tau_b r_n \to  0$ as $n \to \infty$ which is always true if $b=o(n^{4/5})$\end{proof}

\begin{lemma}\label{cor:1}
Assume \ref{ass:s}, \ref{ass:eps}, \ref{ass:K} and \ref{ass:M}. Suppose that $\{a_t\}$, in \ref{ass:eps}, is Normally distributed when $0<\gamma_1\le 1/2$.  Let $\hat s(t)$ be the estimate of $s(t)$ computed on the entire sample (of length $n$). Then $n \to \infty$ and $b=o(n^{4/5})$ implies $\hat{q}_{n,b}(\gamma_2) \convp q(\gamma_2)$ for any $\gamma_2 \in (0,1)$.
\end{lemma}
\begin{proof}
Using the same arguments as in Lemma \ref{prop:2} we have that $\hat{G}_{n,b}(x)-G_{n,b}(x) = o_p(1)$ for each point $x$. By the continuity of  $G(x)$ at all $x$ we have  that  $q_{n,b}(\gamma_2) \convp q(\gamma_2)$ by Lemma \ref{lemma2}. Therefore  $\hat{q}_{n,b}(\gamma_2) \convp q(\gamma_2)$. 
\end{proof}
Note that assumption $b=o(n^{4/5})$ is needed to deal with  \ref{ass:eps}-LRD, however for \ref{ass:eps}-SRD only  we would only need $b=o(n)$.

\paragraph{Proof of Theorem~\ref{th:conv_var_dist}.} Let $P^*(X)$ and $\ex^*(X)$ be the conditional probability and the conditional expectation of a random variable $X$ with respect to a set $\chi = \set{Y_1,\ldots,Y_n}$. Let $\hat G_{n,b_1}^b(x)$ be the same as $\hat G_{n,b}(x)$,  but now $\hat
s(t)$ is estimated on each subsample of length $b$, and the variance of the error term is computed on the same subsample of length $b_1 < b$.  Without loss of generality, we consider the first observaiton with $t=1$ as in Algorithm \ref{algo}. Then,
\begin{equation*}
\frac{1}{b}\sum_{i=1}^b\left(\hat \varepsilon_i-\varepsilon_i\right)^2\convp\text{MISE}(\hat s;h)=O_p\left(\Lambda_b^{-4/5}\right),
\end{equation*}
using Lemma \ref{lemma1} as in the proof of Lemma \ref{prop:2}. Let $b_1=o(b^{4/5})$. 

\REV{
Let $Z_i(x)=\uno\set{\tau_{b_1}\left(\hat V_{n,b_1,i}-V_n\right)\le x}$
and $Z_i^*(x)=\uno\set{\tau_{b_1}\left(\hat V_{n,b_1,I_i}-V_n\right)\le
x}$. $I_i$ is a Uniform random variable on  $I=\set{1,2,\ldots,n-b+1}$.
$P(Z_i^*(x)=Z_i(x)|\chi)=\frac{1}{n-b+1}$ $\forall i$ at  each
$x$. Write 
$\tilde{G}_{n,b_1}(x)=\frac{1}{K}\sum_{i=1}^KZ_i^*(x)$, it follows that 
\begin{displaymath}
E^* \tonde{ \tilde{G}_{n,b_1}(x)} = \frac{1}{n-b+1}
\sum_{i=1}^{n-b+1}Z_i(x) = \hat{G}_{n,b_1}^b(x)\convp G(x),
\end{displaymath}
as $n\rightarrow\infty$, the latter is implied by by Lemma \ref{prop:2}, and the fact that $\tau_{b_1}\Lambda_b^{-4/5}\conv 0$ when $0<\gamma_1\le 1$ in assumption \ref{ass:eps}.   \\%
Since $\{I_i\}$  is the set of uniform random variables sampled without replecement, we can apply Corollary 4.1 of  \cite{Romano1989}. 
Therefore it follows that $\tilde G_{n,b_1}(x)-\hat G_{n,b_1}^b(x)\convp 0$ as $K\rightarrow\infty$ and $n\rightarrow\infty$. Applying the  delta method approach
\[
\tilde G_{n,b_1}(x)-G(x)=\left(\tilde G_{n,b_1}(x)-\hat G_{n.b_1}^b(x)\right)+\left(\hat G_{n.b_1}^b(x)-G(x)\right)\convp 0,
\]
as $K\rightarrow\infty$, $n\rightarrow\infty$ and $\forall x$. Since $G(x)$ is continuous, the convergence is uniform because of  the argument of the last part of the proof of Theorem 2.2.1 in \cite{politis-etal-1999b}.
} 
This concludes  the proof.
\qed

\paragraph{Proof of Corollary \ref{cor:2}.}
The results follow from the proof of Lemma \ref{cor:1} by replacing Lemma \ref{prop:2}  with  Theorem~\ref{th:conv_var_dist}. \qed

\paragraph{Proof of Theorem~\ref{th:convergence_dist_Q}.}
By (\ref{eq:Qn}) we have that
\[
\widehat{SNR} = 10 \log_{10} \left( \frac{\frac{1}{n}\sum_{i=1}^{n} \hat s^2(t_i)}{\frac{1}{m}\sum_{i=1}^m\left(\hat{\eps}_i-\bar{\hat{\eps}} \right)^2} \right)
\]
and $SNR=10\log_{10}\left({\sigma_{\varepsilon}^{-2}{\int s^2(t)dt}}\right)$. First, we analyze the quantity $\tau_m(\widehat{SNR}-SNR)$. So we can write
\[
\tau_m(\widehat{SNR}-SNR)=10\tau_m\left[\log_{10}\left(\frac{\frac{1}{n}\sum_{i=1}^{n} \hat s^2(t_i)}{\int s^2(t) dt}\right)-\log_{10}\left(\frac{\hat V_m}{\sigma_{\varepsilon}^2}\right)\right]=
\]
\[
=-10\tau_m\left[\log_{10}\left(1+\frac{\hat V_m-\sigma_{\varepsilon}^2}{\sigma_{\varepsilon}^2}\right)-\log_{10}\left(1+\frac{\frac{1}{n}\sum_{i=1}^{n} \hat s^2(t_i)-\int s^2(t) dt}{\int s^2(t) dt}\right)\right]=-10\tau_m(I-II).
\]
\REV{
Using the same arguments as in the proof of Theorem~\ref{th:kern_opt}, it follows that $\hat V_m-\sigma_{\varepsilon}^2=O_p(\tau_m^{-1})$.} Expanding $\log_{10}(1+x)$ in Taylor's series, we have that
\[
\tau_mI=\frac{\tau_m}{\sigma_{\varepsilon}^2}\left(\hat V_m-\sigma_{\varepsilon}^2\right)+o_p(1)
\]
and 
\begin{equation}
\label{eq:IIm}
\tau_mII=\frac{\tau_m}{\int s^2(t) dt}\left(\frac{1}{n}\sum_{i=1}^{n} \hat s^2(t_i)-\int s^2(t) dt+o_p(\Lambda_n^{-2/5})\right)=o_p(1).
\end{equation}
\REV{Now, we show the last result. From the proof of Theorem \ref{th:kern_opt} and by assumption \ref{ass:K}, we have that $\hat s(t)=s(t)\left(1+O_p(\Lambda_n^{-2/5})\right)$. Therefore, 
\begin{equation}
\label{eq:s2}
\hat s^2(t)-s^2(t)=\left[\hat s(t)-s(t)\right]\left[\hat s(t)+s(t)\right]=O_p(\Lambda_n^{-2/5}).
\end{equation}
Now, we can write
\[
\frac{1}{n}\sum_{i=1}^{n} \hat s^2(t_i)-\int s^2(t) dt=\left(\frac{1}{n}\sum_{i=1}^{n} \hat s^2(t_i)-\frac{1}{n}\sum_{i=1}^{n} s^2(t_i)\right)+\left(\frac{1}{n}\sum_{i=1}^{n} s^2(t_i)-\int s^2(t)dt\right)=I_s+II_s.
\]
By using the convergence of the quadrature of a bounded and continuous function to its integral, it follows that $II_s=O(n^{-1})$. By (\ref{eq:s2}), we have that 
\[
I_s=\frac{1}{n}\sum_{i=1}^{n} \left(\hat s^2(t_i)-s^2(t_i)\right)=O_p(\Lambda_n^{-2/5}).
\]
Since $m=o(n^{2/5})$, it follows that $\tau_m\Lambda_n^{-2/5}\rightarrow 0$ as $n\rightarrow\infty$. So, (\ref{eq:IIm}) is shown.
}

Hence, we can conclude that $\tau_m(\widehat{SNR}-SNR)$ has the same asymptotic distribution as $\frac{\tau_m}{\sigma_{\varepsilon}^2}\left(\hat V_m-\sigma_{\varepsilon}^2\right)$ by the Slutsky's Theorem. Therefore, assumption 3.2.1 of \cite{politis-etal-1999b} is verified by Theorem \ref{th:conv_var_dist}. 

\REV{Consider the SNR  evaluated at a given point, namely  $SNR_i=10\log_{10}\left(\frac{s^2(t_i)}{\sigma_{\varepsilon}^2}\right)$.}, and  write $\tau_{b_1}\left(\widehat{SNR}_{n,b,I_i}-\widehat{SNR}\right)$ in $\mathbb{Q}_n(x)$ as 
\REV{
\[
\tau_{b_1}\left(\widehat{SNR}_{n,b,I_i}-\widehat{SNR}\right)=\tau_{b_1}\left(\widehat{SNR}_{n,b,I_i}-SNR_{I_i}\right)-\tau_{b_1}\left(\widehat{SNR}-SNR\right)+
\]
\[
+\tau_{b_1}\left(SNR_{I_i}-SNR\right)=S_1-S_2+S_3,
\]
}
for a given subsample starting at $I_i$. By using the first part of this proof, it follows that $S_2=O_p(\tau_{b_1}/\tau_{m})=o_p(1)$ since $b_1/m\rightarrow 0$ when $n\rightarrow\infty$. %
\REV{%
Now, in order to deal with the quantity $S_1$, we need to show that
\begin{equation}
\label{eq:sb}
\frac{1}{b}\sum_{j=i}^{i+b-1}\left[s\left(\frac{j-i+1}{b}\right)\right]^2\rightarrow s^2(t_i)\qquad \mbox{as }\quad n\rightarrow\infty,
\end{equation}
where $t_i$ is the initial point in the block of $b$ values.  By using again the convergence of the quadrature of a bounded and continuous function to its integral, we have that $\frac{1}{b}\sum_{j=i}^{i+b-1}\left[s\left(\frac{j-i+1}{b}\right)\right]^2\rightarrow\int_0^1\left(s^b_i(t)\right)^2dt$ as $n\rightarrow\infty$, $b\rightarrow\infty$ and $b/n\rightarrow 0$. The quantity $s_i^b(\cdot)$ denotes the portion of the signal in the block of $b$ values in $(0,1)$ with $i$ the index for the initial point. Note that  $b/n\rightarrow 0$,  and by  the mean value theorem $\int_0^1\left(s_i^b(t)\right)^2dt\rightarrow s^2(t_i)$.  By using, again, the first part of this proof and by (\ref{eq:sb}), we have that $\tau_{b_1}\left(\widehat{SNR}_{n,b,I_i}-SNR_{I_i}\right)$ has the same asymptotic distribution as $\frac{\tau_{b_1}}{\sigma_{\varepsilon}^2}\left(\hat V_{n,b_1,I_i}-\sigma_{\varepsilon}^2\right)$. Now we  study the quantity $S_3$. First, we show that
\begin{equation}
\label{eq:Nb}
\frac{1}{n-b+1}\sum_{i=1}^{n-b+1}\uno\left\{\tau_{b_1}\left(SNR_i-SNR\right)>x\right\}\rightarrow 0
\end{equation}
when $n\rightarrow\infty$ with some $x>0$. Since $SNR_i-SNR=10\log_{10}\left(\frac{s^2(t_i)}{\int s^2(t)dt}\right)$, the equation in (\ref{eq:Nb}) becomes
\[
\frac{1}{n-b+1}\sum_{i=1}^{n-b+1}\uno\left\{\frac{s^2(t_i)}{\int s^2(t)dt}>10^{\frac{x}{10\tau_{b_1}}}\right\}.
\]
We have that
\begin{equation}
\label{eq:s2R}
\frac{1}{n-b+1}\sum_{i=1}^{n-b+1}s^2(t_i)=\int s^2(t)dt+O(n^{-1}).
\end{equation}
Moreover, $\frac{s^2(t_i)}{\int s^2(t)dt}>10^{\frac{x}{10\tau_{b_1}}}$ can be written as 
\[
\tau_{b_1}\left( \frac{s^2(t_i)}{\int s^2(t)dt}-1\right)>\tau_{b_1}\left(10^{\frac{x}{10\tau_{b_1}}}-1\right).
\]
Summing  over the index $i$ and dividing by $n-b+1$, we can write
\[
 \tau_{b_1}\left( \frac{\frac{1}{n-b+1}\sum_{i=1}^{n-b+1}s^2(t_i)}{\int s^2(t)dt}-1\right)>\tau_{b_1}\left(10^{\frac{x}{10\tau_{b_1}}}-1\right).
\]
Since $\tau_{b_1}\left(10^{\frac{x}{10\tau_{b_1}}}-1\right)\rightarrow c>0$ when $b_1\rightarrow\infty$, by using equation (\ref{eq:s2R}) we obtain 
\[
\tau_{b_1}\left( \frac{\frac{1}{n-b+1}\sum_{i=1}^{n-b+1}s^2(t_i)}{\int s^2(t)dt}-1\right)=O(\tau_{b_1}n^{-1})\rightarrow 0\qquad \mbox{ as }\quad n\rightarrow\infty.
\]
Therefore  $\frac{N_n^b}{n-b+1}\rightarrow 0$ as $n\rightarrow\infty$, where $N_n^b=\sum_{i=1}^{n-b+1}\uno\left\{\frac{s^2(t_i)}{\int s^2(t)dt}>10^{\frac{x}{10\tau_{b_1}}}\right\}$. Then, (\ref{eq:Nb}) is shown.
}
\REV{
As in the proof of Slutsky's Theorem, we  split $\mathbb{Q}_n(x)$ as the sum of three empirical distribution function computed over  $S_1$, $S_2$ and $S_3$ respectively.  
Here the random variables $I_i$ are treated as in the proof of Theorem \ref{th:conv_var_dist}.
} 

\REV{%
Based on the argument above only the component of $\mathbb{Q}_n(x)$ computed over $S_1$ has a non degenerate limit distribution, and this will the same as the asymptotic distribution of the estimator for the variance of the error term.  The proof is now completed.%
} 
\qed



\bibliographystyle{chicago}
\bibliography{REFS}%
\end{document}